%% file: main_coord2016.tex
\documentclass{CSML}
\pdfoutput=1

\usepackage{lastpage}
\newcommand{\dOi}{13(3:5)2017}
\lmcsheading{}{1--\pageref{LastPage}}{}{}%
{Oct.~28, 2016}{Jul.~19, 2017}{}
\usepackage[english]{babel}
\usepackage[utf8x]{inputenc}
\PassOptionsToPackage{hyphens}{url}\usepackage[breaklinks]{hyperref}
\hypersetup{hidelinks}

\usepackage{stmaryrd}
\usepackage{epsfig,moreverb,multirow}
\usepackage{mathrsfs} % for \mathsf{}

\usepackage{verbatim}
\usepackage{url}
\usepackage{graphicx}
\usepackage{xcolor}
\usepackage{amsmath,amssymb,amsfonts,thmtools,thm-restate}
\usepackage{extarrows}
\usepackage{latexsym}

\usepackage{mymacros}
\usepackage{mathpartir}

\usepackage{tikz}                   
\usetikzlibrary{shapes,arrows,positioning}

\allowdisplaybreaks

\def \TirNameStyle #1{\footnotesize #1}

\begin{document}
\title{Tracing where IoT data are collected and aggregated}
\thanks{Partially supported by Universit\`a di Pisa PRA\_2016\_64 Project \emph{Through the fog}.}

\author{Chiara Bodei \and Pierpaolo Degano \and Gian-Luigi Ferrari \and Letterio Galletta}
\address{Dipartimento di Informatica, Universit\`a di Pisa}
\email{\{chiara,degano,giangi,galletta\}@di.unipi.it}


\begin{abstract}

The Internet of Things (IoT) offers the infrastructure of the information society.
It hosts smart objects that automatically collect and exchange data of various kinds, directly gathered from sensors or generated by aggregations.
Suitable coordination primitives and analysis mechanisms are in order to design and reason about IoT systems, and to intercept the implied technological shifts.
We address some of these issues from a foundational point of view.
To study them, we define \IoTLySa, a process calculus endowed with a static analysis that tracks the provenance and the manipulation of IoT data, and how they flow in the system.
The results of the analysis can be used by a designer to check the behaviour of smart objects, in particular to verify non-functional properties, among which security.

\end{abstract}

\maketitle

%\pagestyle{plain}

%%%%
%%%%Wendzel, Steffen. "How to increase the security of smart homes". Communications of the ACM. 59 (5): 47–49. doi:10.1145/2828636. Retrieved 3 September 2016.

\section{Introduction}\label{sec:intro}
	\input{sections/intro.tex}

\color{black}
			
\section{A smart street light control system}\label{sec:example}
	\input{sections/example.tex}

\section{The calculus \IoTLySa}\label{sec:semantics}
	\input{sections/semantics.tex}

\section{Control Flow Analysis}\label{sec:analysis}
	\input{sections/analysis.tex}

\section{CFA at work: verifying security policies}\label{sec:security}
	\input{sections/extensions.tex}

\section{Conclusions}\label{sec:conclusion}
	\input{sections/concl.tex}

\bibliographystyle{splncs03}
\bibliography{biblio}

\appendix
\newpage
\section{ Notation}\label{sec:notation}
	\input{sections/notation.tex}
\section{ Proofs}\label{sec:proofs}
	\input{sections/proof.tex}

%%%%%%%%%%%%%%%%%%%%%%%%%%%%%%%%%%%%%%%%%%%%%%%%%%%%%%%%%%%%%%%%%%%%%%%%

\end{document}

%% file: sections/intro.tex
% !TEX root = ../main_coord2016.tex

Nowadays, an increasingly huge number of heterogeneous devices can be easily plugged into a cyber-physical communication infrastructure, the Internet of Things (IoT).
``Software is eating the world'' is the vivid slogan referring to the {\em smartification} of the objects and devices around us.
The vision offered by IoT as the infrastructure of the information society is fascinating.
It amounts to a global network of things, each with a unique identifier, ranging from light bulbs to cars, equipped with suitable software allowing things to interact each other and coordinate their behaviour. 
Furthermore, smart devices can automatically exchange information of various
kinds gathered from different sources (e.g.~sensors) or generated by aggregating several data sets.

IoT changes the way we interact with our surroundings.
As an example, a smart alarm clock
can drive our coffeemaker to prepare us a cup of coffee in the morning; 
our home automation system turns on our front door light when we arrive at home;
then our smart TV can suggest us some movies for the evening, based e.g.\ on our previous choices.

More connected smart devices and more applications available on the IoT mean more software bugs and vulnerabilities  to identify and fix.
For instance, a bug can cause you to wake up into a cold house in winter or an attacker can enter into your smart TV or 
baby monitor and use it to severely deplete service availability, as it seems to be the case with the recent DDoS attacks reported in the news.%
\footnote{\url{https://www.theguardian.com/technology/2016/oct/22/smart-devices-too-dumb-to-fend-off-cyber-attacks-say-experts}}

Smart devices exhibit and require \emph{open-endedness} to achieve full interactive and cooperative behaviour over the Internet.
Actually, they  generalise the so-called ``embedded systems''  that essentially are controllers of machines {\em not} connected to the Internet and therefore we consider them to live in a \emph{closed} world.
Consequently, new software solutions have emerged for supporting the design and development of IoT, e.g.\ Amazon AWS for IoT and Google Brillo.
We argue that the standard formal techniques and tools need to be adapted in order to support open-endedness of IoT applications and the new complex phenomena that arise in this hybrid scenario.

Here, we contribute to this emerging line of research by introducing the kernel of a formal design framework for IoT, which will provide us with the foundations to develop verification techniques and tools for checking properties of IoT applications.

Our starting point is the process calculus \IoTLySa, a dialect of {\LySa}~\cite{BBDNN_JCS,BNN03}, within the process calculi approach to IoT~\cite{lanese13,Merro_Coord16}.
It has primitive constructs to model sensors and actuators, and suitable primitives both for processing data and for managing the coordination and communication capabilities of interconnected smart objects.
We implicitly assume that sensors are active entities that read data from the physical environment at their own fixed rate.
Actuators instead are passive: they just wait for a command to become active and operate on the environment.
Briefly, our calculus consists of:
\begin{enumerate}
\item systems of nodes, made of (a representation of) the physical components, i.e.\ sensors and actuators, and of software control processes for specifying the \emph{logic} of the objects in the node, including the manipulation of data gathered from sensors and from other nodes.
Intra-node generative communications are implemented through a shared store \`a la Linda~\cite{G85,CG01}.
The adoption of this coordination model supports a smooth implementation of the 
\emph{cyber-physical control architecture}:
physical data are made available by sensors to software entities that manipulate them and trigger the relevant actuators to perform the desired actions of the environment.

\item a primitive for asynchronous multi-party communication among nodes, which can be easily tuned to take care of various constraints, mainly those concerning proximity;

\item functions to process and aggregate data.
\end{enumerate}

\noindent
Our present version of \IoTLySa\ is specifically designed to model monitoring system typical of smart cities, factories or farms.
In this scenario, smart objects never leave their locations, while mobile entities, such as cars or people, carry no smart device and can only trigger sensors.
For this reason we do not address mobility issues here.

%In this scenario, smart objects do not move and for this reason here we do not address mobility issues. 

A further contribution of this paper is the definition of an analysis for \IoTLySa\ to statically predict the run time behaviour of smart systems.
We introduce a  Control Flow Analysis
that safely approximates the behaviour of a system of nodes.
Essentially, it describes the interactions among nodes, tracks how data spread from sensors to the network, and how data are manipulated.

Technically, our analysis abstracts from the concrete values and only considers  their provenance and how they are put together and processed, giving rise to \emph{abstract values}.
In more detail, it returns for each node $\ell$ in the network: 
\begin{itemize}
\item
an abstract store $\widehat{\Sigma}_\ell$ that records for each sensor and each variable a super-set of the abstract values that they may denote at run time; 
\item
a set $\kappa(\ell)$ that over-approximates the set of the messages received by the node $\ell$, and for each of them its sender;
\item
a set $\Theta(\ell)$ of possible abstract values computed and used by the node $\ell$.
%\item
%\textcolor{green}{
%a set $\alpha(\ell,j)$ that includes all the annotated actions ? of the actuator j that may be triggered in the node l.
%}
\end{itemize}
\noindent
The results of the analysis provide us with the basis for checking and certifying various properties of IoT systems.  
As it is, the components $\kappa$ and $\Theta$ track how data may flow in the network and how they influence the outcome of functions.
An example of property that can be statically checked using the component $\kappa$ is the detection of redundant communications, thus 
providing the basis for refactoring the system to increase its performance.
Another example concerns which nodes, if any, use the values read by a specific sensor, and this property simply requires inspecting the component $\Theta$ to be established.

In order to assess the applicability of our analysis for verifying IoT systems, we 
additionally consider some security properties.
For that, we extend our core calculus with cryptographic primitives, and propose a general schema  in which some classical security policies can be expressed, in particular secrecy and access control.
We show then that static checks on the outcome of the analysis help in evaluating the security level of the system and in detecting its potential vulnerabilities.

%%%\textcolor{green}{
%%%Finally, since we can detect which actions of actuators are actually triggered, the analysis 
%%%might suggest to use a simpler actuator if some of its actions are never exercised, or even to remove it if it is never used.
%%%}

\subsection*{Outline of the paper} 
The paper is organised as follows.
The next section intuitively introduces our proposal with the help of an illustrative example.
In Section~\ref{sec:semantics} 
we introduce the process calculus \IoTLySa, and we present our Control Flow Analysis 
in Section~\ref{sec:analysis}.
We consider security issues in Section~\ref{sec:security}.
Concluding remarks and related work are in Section~\ref{sec:conclusion}. 
%All the proofs and some auxiliary results are in the Appendix, %~\ref{sec:proofs},
%together with a table with the abbreviations and the notation used.
The appendixes contain all the proofs with some auxiliary results,
and a table with the abbreviations and the notation used.

Portions of Sections~\ref{sec:example}, \ref{sec:semantics}, and~\ref{sec:analysis} appeared in a preliminary form in~\cite{BDFG_Coord16} where ($i$)  sensors did not probe the operating environment and the effects of actuators where not tracked; and ($ii$) the analysis was less precise, because data were abstracted in a coarser way. 
Section~\ref{sec:security} re-works almost completely our early proposal of checking security policies in~\cite{BDFG_ICE2016}.

%, and the main differences are discussed when relevant.

%% file: sections/example.tex
% !TEX root = ../main_coord2016.tex

The IoT European Research Cluster (IERC) has recently identified  \emph{smart lighting} in smart cities~\cite{IERC} as one of most relevant applications for the Internet of Things. 
Recent studies, e.g.~\cite{Elejoste,ECMC14}, show that smart street light control systems represent effective solutions to improve energy efficiency. 
Many proposed solutions are based on sensors that acquire data about the physical environment and regulate the level of illumination according to the detected events. 
In this section we show how this kind of scenario can be modelled in \IoTLySa\ and what kind of information our Control Flow Analysis provides to designers. 

%Assume that a city council plans to install an intelligent street light control system to optimise energy consumption(see example in~\cite{ECMC14}).
\subsection{System specification.}
We consider a simplified system working on a one-way street, inside a restricted traffic zone.
It is made of two integrated parts.
The first consists of smart lamp posts that are battery powered, can sense their surrounding environment and can communicate with their neighbours to share their views. 
If (a sensor of) the lamp post perceives a pedestrian and there is not enough light in the street it switches on the light and communicates the presence of the pedestrian to the lamp posts nearby.
When a lamp post detects that the level of battery is low, it informs the supervisor of the street lights, $N_s$, that will activate other lamp posts nearby.
The second component of the street light controller uses the electronic access point to the street. 
When a car crosses the checkpoint, a message is sent to the supervisor of the street accesses, $N_a$, that in turn notifies the presence of the car to $N_s$, which acts as a point of connectivity
to the Internet.
A notice is also sent to the node $N_{pd}$ that represents a cloud service of the police department.
This service checks whether the car is enabled to enter that restricted zone, through automatic number plate recognition; below we will omit any further detail on this node, 
e.g.\ the components for fining the driver.
The supervisor $N_s$ sends a message to the lamp post closest to the checkpoint that starts a forward chain till the end of the street, thus completing the specification of the overall cooperative behaviour.
%Such
%heterogeneous sensing provides some form of cooperative behaviour through the activation of a subset
%of the sensors and/or actuators, according to information originating from another subset of sensors.
The structure of our control light system is in~\figurename~\ref{fig:example}.
Below, we will often use a sugared version of the syntax that is made precise in sub-section~\ref{subsect:syntax}.

In our model we assume that each sensor has a unique identifier, hereafter a natural number. 
Analogously, also actuators have a unique identifier.
Since every node has a fixed number of variables, the store of the node can be seen as an array, a portion of which is designated to record the values read by sensors.
A sensor identifier is then used as the index to access its reserved store location. To emphasise that indexes are sensor identifiers, we underline them.

\begin{figure}[t]
\centering
\includegraphics[scale=0.4]{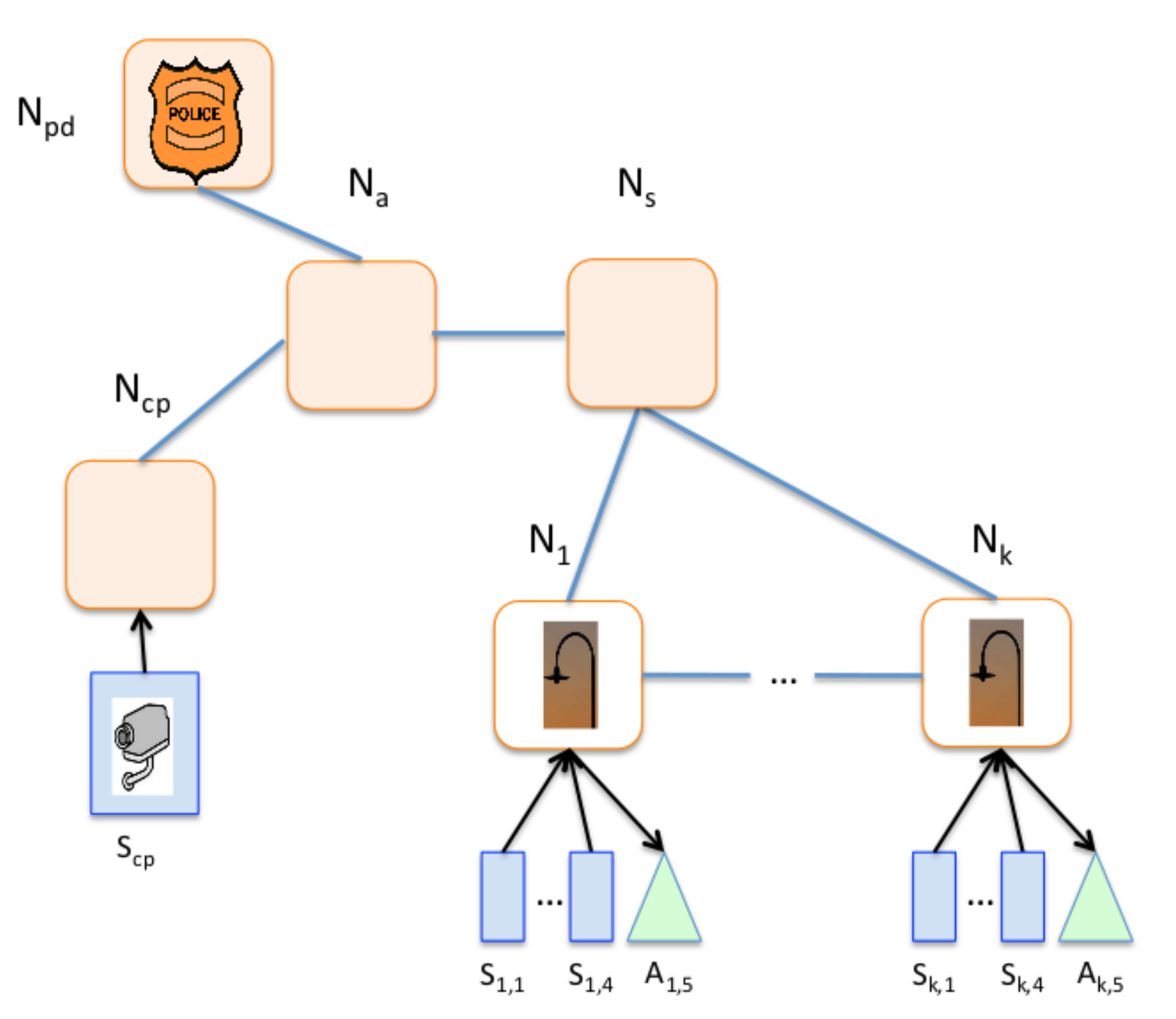}
\caption{The organisation of nodes in our street light control system.}
\label{fig:example}
\end{figure}

In \IoTLySa\, each node consists of control processes, sensors and actuators, and of a local store.
Processes specify the logic of 
the system: they coordinate sensors and actuators, communicate with the other nodes and manage data gathered from sensors and from other nodes.

We now define the checkpoint $N_{cp}$ node.
It only contains a visual sensor $S_{cp}$, defined below, that takes a picture of the car detected in the street:
\[
S_{cp} = \mu h. (\tau.probe(\underline{1})). \tau. \ h
%	S_{cp} = \mu h. (\tau.1 := v_{p}). \tau. \ h
\]
where $\underline{1}$ is the identifier of the sensor $S_{cp}$, and $probe(\underline{1})$ returns the picture of the car. 
The sensor makes the picture available to the other components of the node $N_{cp}$ by storing it in the location (identified by) $\underline{1}$ of the shared store. 
The action $\tau$ denotes internal actions of the sensor, which we are not interested in modelling, e.g.~adjusting the camera focus; the construct $\mu h.$ implements the iterative behaviour of the sensor.
Then, the taken picture is 
enhanced (by  using the function $noiseRed$ for reducing noise) by the process $P_{cp}$ and sent to the supervisor $N_a$ 
\[
	P_{cp} = \mu h. (z := \underline{1}).(z' := noiseRed(z)).\OUTM{z'}{\{\ell_a\}}. \ h
\]
%    where $\ell_a$ is the label of the node $N_a$ and the assignment $z := \underline{1}$ stores in $z$ the picture read by the sensor $S_{cp}$  (recall that $\underline{1}$ is the identifier of $S_{cp}$).
where $\ell_a$ is the label of the node $N_a$; the assignment $z := \underline{1}$ stores in $z$ the picture read by the sensor $S_{cp}$  (recall that $\underline{1}$ is the identifier of $S_{cp}$); and $\langle\!\langle z' \rangle\!\rangle$ is 
a multi-output communication sending $z'$ to the nodes with labels indicated after $\triangleright$ (here only to the single node $N_a$).
The checkpoint $N_{cp}$ is defined as 
\[
	 N_{cp} = \ell_{cp} : [ P_{cp} \ \|\ S_{cp}  \ \|\ B_{cp} ]
\]
where $\ell_{cp}$ is the identifier of $N_{cp}$ and $B_{cp}$ abstracts other components we are not interested in, among which its store $\Sigma_{cp}$.
The node $N_a$ receives the picture and communicates the presence of the car to the lamp posts supervisor $N_s$ and to the police department $N_{pd}$.
The specification of $N_a$ is as follows
\[
 N_{a} = \ell_a : [\ \mu h. \INPS{}{x}{\OUTM{car, x}{\{\ell_s, \ell_{pd}\}.\ h}} \ \|\ B_{a}\ ]
\]
where $\ell_s$ and $\ell_{pd}$ are the identifiers of $N_s$ and $N_{pd}$, respectively (see below for the intuition of the general format of the input $(; x)$).
The supervisor $N_s$ contains the process $P_{s,1}$ that receives the picture from $N_a$ and sends a message to the node closest to the checkpoint, call it $N_1$, labelled with $\ell_1$:
\[
 P_{s,1} = 	\mu h. \INPS{car}{x}{\OUTM{x}{\{\ell_{1}\}}. \ h}
\]
The input $(car; x)$ is performed only if the first element of the corresponding output matches the element before the ``;'' 
(in this case the constant $car$), and the store variable $x$ is bound to the value of the second element of the output (see below for the full definition of $N_s$).
 
In our smart street light control system there is a node $N_p$ for each lamp post, each of which has a unique identifier $p \in [1,k]$.
Each lamp post is equipped with four sensors to sense $(1)$ the environment light, $(2)$ the solar light, $(3)$ the battery level and $(4)$ the presence of a pedestrian. 
We define each of them as follows
%The lamp posts have four sensors to sense $(1)$ the environment light, $(2)$ the solar light, $(3)$ the battery level and $(4)$ the presence of a pedestrian. 
%Each of them is defined as follows
\[
S_{p,i} = \mu h.\, probe(i). \ \tau.\;h
%	S_{p,i} = \mu\,h.\, (i := v). \ \tau.\;h
\]
where $probe(i)$ returns the perceived value by the $i^{th}$ sensor $S_{p,i}$ and $i \in [1,4]$, and stores it.  
%Each sensor writes its value $v$ in a designated memory location $i$, shared with the control processes.
After some internal actions $\tau$, the sensor $S_{p,i}$ iterates its behaviour.
The actuator for the lamp post $p$ is defined as
\[
	A_{5} = \mu h.\, (\!|\underline{5}, \{\mathsf{turnon}, \mathsf{turnoff} \}|\!).\;h
\]
It only accepts a message from $N_c$ whose first element is its (undelined) identifier (here \underline{5}) and whose second element is either command $\mathsf{turnon}$ or $\mathsf{turnoff}$ and executes it. 

The control process of each lamp post node is composed by two parallel processes, $P_{p,1}$  and $P_{p,2}$. 
The first process is defined as follows
%\begin{align*}
%P_{i1} & = \Pi.C \\
%\Pi    & = x_1:= 1. x_2:= 2. x_3:= 3. x_4:= 4\\
%C & = (x_4 == true)\ ?\  (x_1 \leq th_1 \land x_2 \leq th_2)\ ?\ \\
% &  \qquad \qquad \qquad \qquad\  (x_3 \geq th_3)\ ?\ \OUTS{5,\mathsf{turnon}}{\OUTM{x_4}{L_i}}\\
% &  \qquad \qquad \qquad \qquad \qquad \qquad \quad : \ \OUTM{\mathsf{err}}{L_e} \\
%  &  \qquad \qquad \qquad \qquad \qquad \qquad \quad : \ 0 \\
%    &  \qquad \qquad \qquad \quad\ \ : \  \OUTS{5,\mathsf{turnoff}}{}
%\end{align*}
\begin{align*} 
P_{p,1} = \mu h. & (x_1:= \underline{1}.\, x_2:= \underline{2}.\, x_3:= \underline{3}.\, x_4:= \underline{4}). \\
& (x_4 = true)\ ?\  \\
& \qquad (x_1 \geq th_1 \land x_2 \geq th_2)\ ?\ \\
 &  \qquad \qquad \qquad \qquad\ \  \ (x_3 \geq th_3)\ ?\ \OUTS{\underline{5},\mathsf{turnon}}{\OUTM{x_4}{L_p}}.\ h\\
 &  \qquad \qquad \qquad \qquad \qquad \qquad \quad \ \ : \ \OUTM{\mathsf{err}, \ell_p}{\{\ell_s\}}. \ h \\
  &  \qquad \qquad \qquad \qquad \qquad \qquad : h \\
    & \ \ \ \ \qquad \qquad: \  \OUTS{\underline{5},\mathsf{turnoff}}{h}
\end{align*}
The process reads the current values from the sensors and stores them into the local variables $x_i$. 
The nested conditional statement says that the actuator is turned on if (i) a pedestrian is detected in the street ($x_4$ holds), (ii) the intensity of environment and solar lights are greater than or equal to the given thresholds 
$th_1$ and $th_2$, and (iii) there is enough battery (at least $th_3$).
In addition, the presence of the pedestrian is communicated to the lamp posts nearby, whose labels, typically $\ell_{p-1}$ and $\ell_{p+1}$, are in $L_p$.
Instead, if the level battery is insufficient, an error message, including its identifier $\ell_p$, is sent to the supervisor node, labelled $\ell_s$.
The second process $P_{p,2}$ is defined as follows
\[
\! P_{p,2} = \mu h. \INPS{}{x}{(x = true \lor is\_a\_car(x)) \ ? \ (\OUTS{\underline{5},\mathsf{turnon}}{\OUTM{x}{L_p}). h}:  \OUTS{\underline{5},\mathsf{turnoff}} h}
\]
It waits for messages from its neighbours or from the supervisor node $N_s$.
When one of them is notified the presence of a pedestrian ($x = true$) or of a car ($is\_a\_car(x)$ holds), the current lamp post orders the actuator to switch the light on.
Each lamp post $p$ is described as the {\sc IoT}-{\sc LySa} node below
\[
N_p = \ell_p : [\Sigma_p \ \|\ P_{p,1} \ \|\ P_{p,2} \ \|\ S_{p,1} \ \|\ S_{p,2} \ \|\ S_{p,3} \ \|\ S_{p,4} \ \|\ A_{p,5}]
\]
where $\Sigma_p$ is the store of the node $\ell_p$, shared among its components. 
%We assume that the locations $i \in [1,4]$ in $N_p$ are reserved to store data from sensor $i$.
%Note that the control processes $P_{p,1}$ and $P_{p,2}$ are made iterative through the ! operator.
%
The supervisor node $N_s$ of lamp posts is defined below, where $x$ ranges over the $k$ lamp posts
\[
N_s = \ell_s : [\mu h.\, \INPS{err}{x}{\OUTM{true}{L_x}}. \ h\, \ \|\ P_{s,1}  \ \|\ B_s]
\]
where $P_{s,1}$ is the process previously defined. 
As above the input $(err; x)$ is performed only if the first element of the corresponding output matches the constant $err$, and the store variable $x$ is bound to the value of the second element of the output, i.e.\ the label of the relevant lamp post.
If this is the case, after some internal computations, $N_s$ warns the lamp posts nearby $x$ (included in $L_x$) of the presence of a pedestrian.

The whole intelligent controller $N$ of the street lights is then described as the parallel composition of the checkpoint node $N_{cp}$, the supervisors nodes $N_a$ and $N_s$, the nodes of lamp posts $N_p$, with $p \in [1,k]$, and the police department node $N_{pd}$:
\[
N = N_{cp} \mid N_a \mid N_s \mid N_1 \mid \dots \mid N_k \mid N_{pd}
\]

\subsection{Checking properties.}
We would like to statically predict how the system behaves at run time. 
In particular, we want to compute:
$(i)$ how nodes interact each other; 
$(ii)$ how data spread from sensors to the network (tracking); and
$(iii)$ which computations each node performs on the received data.
To do that, we define a Control Flow Analysis, which abstracts from the concrete values by only considering their provenance and how they are manipulated. 
For example, consider the picture sent by the camera of $S_{cp}$ to its control process $P_{pc}$. 
In the analysis we are only interested in tracking where the picture comes from, and not in its actual value; 
so we use (a suitable representation of) the  abstract value $1^{\ell_{cp}}$ to record the camera that took it. 
The process $P_{pc}$ reduces the noise in the pictures and sends the result to $N_a$. 
Our analysis keeps track of this manipulation through (a representation of) the abstract value $noiseRed^{^{\ell_{cp}}}(1^{\ell_{cp}})$, 
meaning that the function $noiseRed$, computed by the node $\ell_{cp}$, is applied to data coming from the sensor with identifier $1$ of $\ell_{cp}$.

In more detail, our analysis returns for each node $\ell$ in the network: 
an abstract store $\widehat{\Sigma}_\ell$ that records for each variable a super-set of the abstract values that it may denote at run time and for each sensor a specific abstract value as exemplified above; 
a set $\kappa(\ell)$ that approximates the set of the messages received by the node $\ell$; 
and the set $\Theta(\ell)$ of possible abstract values computed and used by the node $\ell$.
Actually, abstract values may grow unbounded and we will thus represent them through regular tree grammars.
In our example, for each lamp post labelled $\ell_p$,
the analysis returns in $\kappa(\ell_p)$ both (the grammar representing) the abstract value $noiseRed^{^{\ell_{cp}}}(1^{\ell_{cp}})$ and the sender of that message, i.e. $\ell_{p+1}$.

The result of our analysis can be exploited to perform several verifications. 
For instance, since the pictures of cars are sensitive data, one would like to check whether they are kept confidential.
A simple check on $\kappa$ suffices for discovering that our system fails to protect the pictures because they are always sent in clear.   
An obvious solution is to encrypt these sensitive data before sending them.
However, privacy is not guaranteed either, because the encrypted information is correctly sent to the supervisor of the street access $N_a$ and to the police department $N_{pd}$, but also to all lamp posts.
By inspecting  $\kappa$ and $\Theta$ we detect this violation.
%However, the photo has to be sent to the police department, in order to detect if the car is violating the traffic rules.
A possible solution is making the picture anonymous through a suitable function $an$, e.g.\ blurring the plate, before 
sending it to the supervisor of the street light $N_s$.
In order to reach a balance between
protecting data for privacy and using them for safety, we can resort to the following amended 
code of the process $P_{cp}$ and of the node $N_a$: 
%becomes:
\[
	P_{cp}' = \mu h. (z := \underline{1} ).(z' := noiseRed(z)).\OUTM{\{z'\}_k}{\{\ell_a\}}. \ h
\]
\[
 N_{a}' = \ell'_a : [\ \mu h. \INPS{}{y}{\DECSO{y}{}{x}{k}{}{
 \OUTM{car, \{x\}_{k'}}{\{\ell_{pd}\}}.\ \OUTM{car, an(x)}{\{\ell_s\}.\ h}
                                                                                          }
                                                          } \ \|\ B_{a}\ ]
\]

Another way of exploiting the results of our analysis is detecting whether there are redundant communications, which are possibly power consuming.
For example, since the street is one-way, when a car is present the lamp post at position $p$ needs not to alert the one at $p-1$.
By inspecting $\kappa$ it is easy to detect a redundant, useless communication from the next lamp post.
On this basis, the designer can remove the label $\ell_{p-1}$ from the set $L_p$ of receivers in the definition of $P_{p,2}$ for all lamp posts.

A further issue concerns the fact that an IoT system should be robust enough and work in presence of a faulty device.
Our analysis can support a ``what-if'' reasoning.
For example, suppose that the designer wants to check what happens if a lamp post, say $\ell_{p-1}$, goes out of order.
By inspecting the $\kappa$ component, the designer discovers that $\ell_p$ may receive messages from $\ell_{p-1}$ and $\ell_{p+1}$.
Assuming that redundancy has been removed as sketched above, no message will be sent to $\ell_p$ from $\ell_{p-1}$ when a car is in the street (as specified in the process $P_{i-1, 2}$).
Consequently, all the lamp posts $\ell_q$ with $q \geq p-1$, will not be switched on, even though a car is running in the street.
Clearly, when pedestrians are walking there, the relevant lamp post will be switched on as soon as their presence is detected, except for the faulty one, of course.
An easy fix is enlarging the set $L_p$ in the process $P_{p,2}$ to also contain $\ell_{p+2}$, so introducing a possibly useful redundancy; similarly, the process $P_{s,1}$ will send its message to both $\ell_1$ and $\ell_2$.

%% file: sections/semantics.tex
% !TEX root = ../main_coord2016.tex

The IoT applications operate in a cyber-physical world, and therefore modelling them requires taking into account  both logical and physical aspects. 
Physical data, typically collected by environmental sensors, influence the logics of an application, which in turn modifies the physical world, through actuators.
Here we are only interested in modelling the logical components, in identifying their frontier with the physical world and in abstractly representing the interactions between them.
Typical representations of the world are based on continuous models, like, e.g.~(ordinary or stochastic) differential equations, or on discrete version of them or on hybrid automata~\cite{Henzinger96}.
Instead, for us the physical world is a black box from which sensors, that are always active, can measure the value of some observables with a certain rate of their own.
These values are made available to controllers that trigger suitable actuators.
Actuators are passive entities that can only execute their task when activated;
they operate on the physical world,  so changing the value of some observables.
Furthermore, we assume that the world can autonomously evolve, and that the changes of its state are revealed and detected by sensors.

In order to model the logical components of IoT applications, we adapt the {\LySa} calculus \cite{BBDNN2003lysa,BBDNN_JCS,BNN03}, based 
on the $\pi$-~\cite{MPW92} and the Spi-calculus \cite{AG-spi}. 
 For that we introduce:
 \begin{enumerate}[label={(\roman*)},leftmargin=1.1\parindent]
\item systems of nodes, in turn consisting of sensors, actuators and control processes, plus a shared store within each node, allowing for internal communications;
\item primitives for measuring values in the world with sensors, and for triggering actuator actions;
\item an asynchronous multi-party communication modality among nodes, subject to constraints, mainly concerning
physical proximity;
\item functions to process data;
\item explicit conditional statements.
\end{enumerate}
\noindent
We also extend  our proposal in Section~\ref{sec:security} with 
\begin{enumerate}[resume*]
\item encryption and decryption constructs to represent and handle some logical aspects of security.
\end{enumerate}

%%%%%%%%%%%
\subsection{Syntax.}\label{subsect:syntax}
%%%%%%%%%%%
%\noindent
%{\bf Syntax}
The logical components of an IoT system are specified using a two-level syntax, one describing the whole system and the other its components.
At the first level one defines system, consisting of the parallel composition of a fixed number of labelled nodes.
Each node is defined at the second level, and it hosts a store, control processes, sensors and actuators.
The label $\ell$ uniquely identifies the node $\ell: [B]$ and may represent further characterising information (e.g.\ its location or other contextual information).
%Finally, the operator $|$ describes a system of nodes obtained by parallel composition.
The syntax of systems is as follows.
\[
\begin{array}{ll@{\hspace{2ex}}l}
{\mathcal N} \ni 
N ::= & {\it systems \ of \ nodes} &\\
& \NIL                       & \hbox{empty system} \\
& \ell: [B]                      & \hbox{single node}\;  (\ell \in {\mathcal L} \text{, the set of labels})
\\
& N_1\ |\ N_2 &  \hbox{parallel composition of nodes}
\end{array}
\]
\[
\begin{array}{ll@{\hspace{2ex}}l}
{\mathcal B} \ni 
B ::= & \text {\it node components} &\\
&  \Sigma_\ell  & \hbox{store of node } \ell
\\
& P  & \hbox{process}
\\
& S   & \hbox{sensor, with a unique identifier $i \in \mathcal{I}_\ell$}
\\
& A     & \hbox{actuator, with a unique identifier $j \in \mathcal{J}_\ell$}
\\
& B \ \|\ B & \hbox{parallel composition of node components}
\end{array}
\]

%A node component $B$ contains some sensors $S$ (less than $n = \#(\mathcal{I}_{\ell})$), and some actuators $A$ (less than $m = \#(\mathcal{J}_{\ell})$).
A node component $B$ contains $\#(\mathcal{I}_{\ell})$ sensors $S$ and $\#(\mathcal{J}_{\ell})$ actuators $A$.
It also has finitely many control processes $P$ that use a finite set of variables ${\mathcal X}_\ell$.
We impose that in $\ell: [B]$  there always is a \emph{single} store $ \Sigma_{\ell} : \mathcal{X} _\ell \cup \mathcal{I}_{\ell} \ \rightarrow \mathcal{V} $, where ${\mathcal V}$ is the denumerable set of values, including numbers, booleans etc., but neither labels nor identifiers of sensors and actuators.
(We feel free to omit from here afterwards the label $\ell$ when immaterial.)
Therefore, a store is essentially an array of fixed dimension, and intuitively a variable $x \in {\mathcal X}_\ell$ and an identifier $i \in \mathcal{I}_{\ell}$ are interpreted as indexes in the array (no need of $\alpha$-conversions).
We assume that store accesses are atomic, e.g.\ through CAS instructions~\cite{CAS91}.

The syntax of control processes is as follows
\[
\begin{array}{lll}
%{\mathcal P} \ni 
P ::= & {\it control \ processes} &\\
& \NIL                       & \hbox{inactive process} \\
& \OUTM{E_1, \cdots, E_r}{L}.\,P & \hbox{asynchronous multi-output L$\subseteq {\mathcal L}$} \\
& \INPS{E_1,\cdots,E_j}{x_{j+1},\cdots,x_r}{P}\  & \hbox{input (with matching)}\\
& E?P:Q &  \hbox{conditional statement} \\
& h   &  \hbox{iteration variable}
\\
&\mu h. \ P & \hbox{iteration}
  
\\[.2ex]
& x := E.\,P & \hbox{assignment to $x \in {\mathcal X}_\ell$}
\\
& \OUTS{j, \gamma}{P}& \hbox{output of action $\gamma$ to actuator $j$}
\end{array}
\]

%The process $\NIL$ represents the inactive process.
The prefix $\OUTM{E_1, \cdots, E_r }{L}$ implements a simple form of multicast communication among nodes: the tuple $E_1, \dots, E_r$ is asynchronously sent to the nodes with labels in $L$ and that are ``compatible'' (according, among other attributes, to a proximity-based notion and to the transmission capability of the sender).
The input prefix $(E_1, \!\cdots\!,E_j; x_{j+1},\! \cdots \!,x_r)$
is willing to receive a $r$-tuple, provided that its first $j$ terms match the input ones, and then binds the remaining store variables (separated by a ``;'') to the corresponding values (see \cite{BNN03,BBF15} for a more flexible choice).
Otherwise, the $r$-tuple is not accepted.
A process repeats its behaviour, when defined through the iteration construct $\mu h.\ P$, where $h$ is the iteration variable;
an obvious sanity requirement is that any recursion variable $h$ only occurs within the scope of its binding $\mu h$.
Finally, a process can command an actuator to perform an action over the physical world.

Sensors and actuators 
%(uniquely labelled on ${\mathcal L}$) 
have the form:

\[
\begin{array}{ll@{\hspace{2ex}}l ll@{\hspace{2ex}}l}
%{\mathcal S} \ni 
S &::=  {\it sensors} & &
%{\mathcal A} \ni 
A &::= {\it actuators} &\\
& \NIL &  \hbox{\hspace{-1mm}inactive sensor} &&
 \NIL &  \hbox{\hspace{-1mm}inactive actuator}
\\
& \tau.S & \hbox{\hspace{-1mm}internal action} &
& \tau.A & \hbox{\hspace{-1mm}internal action} 
\\
%& i :=  v.\,{S} & \hbox{\hspace{-1mm}store of $v \in {\mathcal V}$} 
& probe(i).\,{S} & \hbox{\hspace{-.8mm}sense a value by} 
%in $i \in {\mathcal I}_{\ell}$}
%by the $i^{th}$ sensor} 
&& 
(\!|j, \Gamma|\!).\,A & \hbox{\hspace{-1mm}command for actuator $j$} 
\\
&& \hbox{\hspace{-1mm}the $i^{th}$ sensor}
 &&
\gamma.A & \hbox{\hspace{-1mm}triggered action ($\gamma \in \Gamma$)}
%, the set of actions of $A$})

\\
& h & \hbox{\hspace{-1mm}iteration variable} 
&&
h & \hbox{\hspace{-1mm}iteration variable} 
\\
& \mu h\,.\, S & \hbox{\hspace{-1mm}iteration}
&& \mu h\,.\, A & \hbox{\hspace{-1mm}iteration}
\end{array}
\]

\noindent
We recall that each sensor and each actuator is identified by a unique identifier belonging to the sets ${\mathcal I_\ell}$ and ${\mathcal J_\ell}$, respectively.
A sensor can perform an internal action $\tau$, e.g.\ for resetting or for changing its battery mode.
In addition, it senses the physical world and stores the value observed in its location $i$.
Recall that each sensor is dedicated to measure a specific observable, e.g.\ the temperature or the humidity.
Also an actuator can perform an internal action $\tau$.
More interestingly, upon receiving a command $\gamma$ from a controller, an actuator executes $\gamma$, possibly causing a change in the physical state of the world.
Both sensors and actuators can iterate their behaviour.

The syntax of terms follows.

%Finally, the syntax of terms is as follows:
\[
\begin{array}{ll@{\hspace{2ex}}l}
%{\mathcal E} \ni 
E ::= & 
       {\it terms}&\\
& v  & \hbox{value } (v \in {\mathcal V})\\
& i  & \hbox{sensor location } (i \in {\mathcal I_\ell})\\
& x  & \hbox{variable } (x \in {\mathcal X}_\ell)\\
%& z  & \hbox{sensor's variable } (z \in {\mathcal Z})\\
%& \{E_1,\cdots,E_r\}_{E_0}
%              & \hbox{encryption}\; (k\geq 0)
%              \\
& f(E_1, \cdots, E_r) &  \hbox{function on data }  (f \in {\mathcal F})\\
\end{array}
\]
The term $f(E_1, \cdots, E_r)$ is the application of function $f$ to $n$ arguments; 
we assume as given a set of primitive functions ${\mathcal F}$, typically for aggregating or comparing values, be they computed or data sampled from the environment.

%\textcolor{red}{To facilitate our analysis in Section 3, we assume that all the variables are distinct.
%Furthermore, we assume that the variables occurring in the bound names of inputs and write commands are all distinct.}

%%%%%%%%%%
\subsection{Operational Semantics.}
%%%%%%%%%%

Our reduction semantics relies on an abstract model of the evolution of the physical world from a state $\mathcal{E}$ to a state $\mathcal{E}'$.
We write $\mathcal{E} \triangleright \mathcal{E}'$ when a state transition occurs, without detailing how this happens, because for us the world is a black box.
Similarly, we will write $\gamma(\mathcal{E}) \triangleright \mathcal{E}'$ to represent the state transition caused by an actuator that performs the action $\gamma$ on $\mathcal{E}$.

We take the syntactic elements of \IoTLySa\ up to the following \emph{structural congruence} $\equiv$ on nodes, processes, sensors and actuators.
It is standard except for the last rule that equates a multi-output with no receivers and the inactive process, and for the fact that inactive components of a node are all coalesced.

%\begin{table*}[t]
\[
\begin{array}{ll}
-\  & ({\mathcal N}/_{\equiv}, \mid, \NIL)  \text{ and } ({\mathcal B}/_{\equiv}, \|, \NIL)  \mbox{ are commutative monoids} 
\\
%N \mid \NIL \equiv N   &
%N_1 \mid N_2 \equiv N_2 \mid N_1  &
%(N_1 \mid N_2) \mid N_3 \equiv N_1 \mid ( N_2  \mid  N_3) \\
%B \| \NIL \equiv B   &
%B_1 \| B_2 \equiv B_2 \| B_1  &
%(B_1 \| B_2) \| B_3 \equiv B_1 \|( B_2 \| B_3) \\
-\ &\mu \,h\,.\, X \equiv X\{\mu h\,.\, X/h\}  \quad
 \text{ for } X \in \{P, A, S\} \hspace{3cm} (\bigstar)
\\
-\  & \langle \langle E_1,\cdots,E_r \rangle \rangle : \emptyset . \ \NIL\equiv \NIL
\end{array}
\]
%\caption{Congruence rules of \IoTLySa.}
%\label{fig:struct-cong-l}
%\end{table*}

We have a two-level {\em reduction relation\/}  \sfreccia\ reflecting the two-level structure of \IoTLySa.
It is defined as the least relation on both nodes and their components, satisfying the set of inference rules in \tablename~\ref{opsem}.
The arrow is decorated with $\mathcal{E}$, the state of the world before the transition, and $\mathcal{E}'$ after its occurrence.
We assume the standard denotational interpretation $\dsem{E}_\Sigma$ for evaluating terms.

\begin{table*}
\footnotesize
\begin{mathpar}
{
%\inferrule[(S-store)]%
%{ }%
%{\Sigma \parallel i := v.\,{S_i \parallel B
%               \nfreccia
%               \Sigma\{v/{i}\} \parallel S_i \parallel B}
%}%

\inferrule[(Sense)]%
{ measure(i, \mathcal{E}) = v}%
{\Sigma \parallel probe(i).\,{S \parallel B
               \nfreccia
               \Sigma\{v/{i}\} \parallel S \parallel B}
}%

\inferrule[(Asgm)]%
{\dsem{E}_\Sigma = v}%
{ \Sigma \parallel  x:=E.\,P \parallel B \ 
               {\nfreccia} \ 
             \Sigma\{v/x\}  \parallel  P \parallel B}
}             

\inferrule[(Ev-out)]%
{\bigwedge_{i=1}^r  \dsem{E_i}_{\Sigma} = {v_i} }%
{\Sigma  \parallel  \OUTM{E_1,\cdots,E_r}{L}.\,P \parallel B \ \nfreccia \ 
  \Sigma  \parallel  \OUTM{v_1,\cdots, v_r}{L} . \NIL\parallel P \parallel B}

\inferrule[(Multi-com)]%
{\ell_2 \in L \wedge \ Comp(\ell_1,\ell_2)  \wedge \  \bigwedge_{i=1}^{j} \dsem{E_i}_{\Sigma_2} = {v_i} }%
{
 \ell_1: [\OUTM{v_1,\cdots,v_r}{L}. \, \NIL \parallel B_1] \mid
             \ell_2:[\Sigma_2 \ \| \ (E_1,\cdots,E_j;x_{j+1},\cdots,x_r).Q \parallel B_2]
              \\\\           \nfreccia          \\\\
              \ell_1: [\OUTM{v_1,\cdots,v_r}{L \setminus \{\ell_2\}}. \, \NIL \parallel B_1]
             \mid \ell_2: [\Sigma_2\{v_{j+1}/x_{j+1},\cdots,v_r/x_r\} \parallel Q \parallel B_2]
}

{
\inferrule[(Cond)]%
{\dsem{E}_\Sigma = \tt{b}_i \ \ \ i \in\{1,2\}}%
{\Sigma \parallel E?\,P_1:P_2 \parallel B \ 
              \nfreccia \ 
                            \Sigma  \parallel  P_i \parallel B}
 \ \text{with } \tt{b}_1 = \tt{true}, \tt{b}_2 = \tt{false}
 
 \inferrule[(Int)]%
{ }%
{ \tau.\,X \ \nfreccia\  X}
}

\inferrule[(A-com)]%
{\gamma \in \Gamma}%
{\OUTS{j,\gamma}{P \parallel   (\!|j,\Gamma|\!).\,A \parallel B }
             \nfreccia
             \ P \parallel \gamma.\,A \parallel B}

\inferrule[(Act)]%
{\gamma(\mathcal{E}) \triangleright \mathcal{E'}}%
{\gamma.A \sfreccia\  A}

\inferrule[(Phys)]%
{\mathcal{E} \triangleright \mathcal{E'}}%
{N \sfreccia N}

%%%%%% regola per la ricorsione
%%%%%%%%%%%%%%%%%%%%%%%%%%%%
%%%%%%\textcolor{red}{
%%%%%%\inferrule[(Iter)]%
%%%%%%{Z \{\mu \,h\,.\, Z/h\} \sfreccia Z'}%
%%%%%%{\mu \,h\,.\, Z \sfreccia Z'}
%%%%%%                    }

{
\inferrule[(Node)]%
{B \ \sfreccia \ B'}%
{\ell: [B] \ \sfreccia \ \ell:[B']}

\inferrule[(ParN)]%
{N_1 \sfreccia N'_1}%
{{N_1 \mid N_2} \sfreccia {N'_1 \mid N_2}}

\inferrule[(ParB)]%
{B_1 \sfreccia B'_1}%
{B_1 \parallel B_2 \sfreccia B'_1\parallel B_2}

\inferrule[(CongrY)]%
{Y_1' \equiv Y_1 \sfreccia Y_2 \equiv Y'_2}%
{Y'_1 \sfreccia Y'_2}

}

\end{mathpar}

\caption{Reduction semantics, where $X \in \{S, A\}$, 
%\textcolor{red}{
%$Z \in \{P, S, A\}$
%}
 and $Y \in \{N, B\}$.}
\label{opsem}
\end{table*}

The first two rules implement the (atomic) asynchronous update of shared variables inside nodes, by using
the standard notation $\Sigma\{-/-\}$ for store update.
According to (Sense), the $i^{th}$ sensor measures the value $v$ of its observable through the semantic function $measure$ that operates on the environment; then it stores $v$ into the location $i$. 
The rule (Asgm) specifies how a control process updates the variable $x$ with the value of $E$.

The rules (Ev-out) and (Multi-com) drive asynchronous multi-communications among nodes.
In the first a node sends a tuple of values $\mess{v_1,...,v_r}$, obtained by the evaluation 
of $\mess{E_1,...,E_r}$.
Asynchrony is realised by spawning the new process $\OUTM{v_1,\cdots,v_r}{L}. \, \NIL $ (with the inactive process as 
continuation) in parallel with the continuation $P$;
the new process offers the message to all the receivers belonging to the set $L$.
In the rule (Multi-com), the message coming from $\ell_1$ is received by a node labelled $\ell_2$. 
The communication succeeds, provided that (i) $\ell_2$ belongs to the set $L$ of possible receivers, 
%(ii) the two nodes are compatible according to the compatibility function $Comp$,
(ii) the receiver is within the transmision range of the sender, according to the compatibility function $Comp$,
and (iii) that the first $j$ values match the evaluations of the first $j$ terms in the input. 
Moreover, the label $\ell_2$ is removed by the set of receivers $L$ of the message. 
The spawned process terminates when all its receivers have received the message (see the last congruence rule in $\bigstar$).
The role of the compatibility function $Comp$ is crucial in modelling real world constraints on communication.
A common requirement is that inter-node communications are proximity-based, i.e. that only nodes that are in the transmission range of the sender can read the message. This is easily encoded here by defining a possibly non-symmetric predicate (over node labels) yielding true only when the second node is in the transmission range of the first. 
Of course, this function could be enriched in order to consider finer notions of compatibility expressing various policies, e.g.\ topics for event notification.
%Note that if \emph{Comp} varies along time, we recover a simple way of expressing dynamic network topologies.

According to the evaluation of the expression $E$, the rule (Cond) says that the process continues as $P_1$ (if $\dsem{E}_{\Sigma}$ is true) or as $P_2$ (otherwise).
The rule (Int) applies to sensors and actuators and simply governs the occurrence of internal actions.

A process commands the $j^{th}$ actuator through the rule (A-com), by sending it the pair $\langle j, \gamma \rangle$.
The effect is that $\gamma$ prefixes the actuator, if $\gamma$ is one of its actions.
The rule (Act) says that the actuator performs the action $\gamma$ over the current state $\mathcal{E}$ of the environment, triggering an evolution to a new state. 
The other rule that changes the environment state is (Phys).
It models the fact that environments can evolve independently of the application in hand.
Since for us an environment is a black box, we resort to non-determinism to abstract from this kind of independence.

The rules  (ParN) and (ParB) propagate reductions across parallel composition; the rule (Node) lifts the transitions from components to nodes. 
Finally, the rules (CongrY), with $Y \in \{N,B\}$, for nodes and components are the standard reduction rules for the congruence defined in ($\bigstar$).

Back to our example of Section~\ref{sec:example}, 
consider a run where a picture of a car is taken by the camera $S_{cp}$ in the node $\ell_{cp}$ and is sent to the node $\ell_{a}$.
First of all, we briefly recall the definition of the camera $S_{cp}$ and of the control process $P_{cp}$ of $\ell_{cp}$:
\begin{align*}
S_{cp} & =  \mu h. (\tau.probe(\underline{1})). \tau. \ h \\
P_{cp} & =  \mu h. (z := \underline{1}).\underbrace{(z' := noiseRed(z)).\OUTM{z'}{\{\ell_a\}}. \ h}_{P'_{cp}}
\end{align*}
The sequence of transitions is as follows:
\begin{align}
N_{cp} =\ 
& \ell_{cp} [ P_{cp} \parallel S_{cp}  \parallel \Sigma_{cp}  \parallel B_{cp} ] \nfreccia^{\!\!\!\!+} 
\\
& \ell_{cp} [ P_{cp} \parallel \underbrace{\tau. S_{cp}}_{S'_{cp}}  \parallel \underbrace{\Sigma_{cp}\{v / 1\}}_{\Sigma^1_{cp}}  \parallel B_{cp} ] \nfreccia^{\!\!\!\!+} 
\\
& \ell_{cp} [ P'_{cp} \parallel S'_{cp}  \parallel \underbrace{\Sigma^1_{cp}\{v / z\}}_{\Sigma^2_{cp}}  \parallel B_{cp} ] \nfreccia
\\
& \ell_{cp} [ \OUTM{z'}{\{\ell_a\}}. \ P_{cp} \parallel S'_{cp}  \parallel \underbrace{\Sigma^2_{cp}\{v' / z'\}}_{\Sigma^3_{cp}}  \parallel B_{cp} ] \nfreccia
\\
& \ell_{cp} [\underbrace{P_{cp} \parallel S'_{cp}  \parallel \Sigma^3_{cp} \parallel B_{cp}}_{B'_{cp}} \parallel  \OUTM{v'}{\{\ell_a\}}. \ 0 ]  = N'_{cp}
\end{align}
\\
where as usual $\nfreccia^{\!\!\!\!+}$ stands for the occurrence of one or more transitions.
The rules (CongrB), (Int) and (Sense) are applied to move from (3.1) to (3.2) and the value $v$ amounts to $measure(1, \mathcal{E})$; 
the rules (CongrB) and (Asgm) drive the transition from (3.2) to (3.3); 
once again the rule (Asgm) drives the transition from (3.3) to (3.4) where the value $v'$ is the result of applying the function $noiseRed$ to the value $v$;
in the transition from (3.4) to (3.5) the rule (Ev-out) is applied.

As a further example, consider the steps carried out by the node $\ell_a$ to receive the car picture and to forward it to the policy station $\ell_{pd}$ and to the lamp post supervisor $\ell_{s}$.
Here we recall the definition of the node $N_a$:
\[
	N_a = \ell_a [\ \underbrace{\mu h.(;x).\OUTM{car, x}{\{\ell_s, \ell_{pd}\}.\ h}}_{P_a}  \parallel  \Sigma_{a} \parallel B_{a}\ ]
\]

The transitions are as follows:
\begin{align}
& N'_{cp} \mid N_a \nfreccia 
\\
& \ell_{cp} [ B'_{cp} \parallel \OUTM{v'}{\{\ell_a\}}. \ 0 ] \mid 
\ell_a [\ (;x).\OUTM{car, x}{\{\ell_s, \ell_{pd}\}.\ P_a}  \parallel  \Sigma_{a} \parallel B_{a}\ ] \nfreccia 
\\
& \ell_{cp} [ B'_{cp} \parallel \OUTM{v'}{\emptyset}. \ 0 ] \mid 
\ell_a [\ \OUTM{car, x}{\{\ell_s, \ell_{pd}\}.\ P_a}  \parallel  \underbrace{\Sigma_{a}\{v'/x\}}_{\Sigma'_a} \parallel B_{a}\ ] \nfreccia 
\\
& \ell_{cp} [ B'_{cp} ] \mid 
\ell_a [\ \OUTM{car, x}{\{\ell_s, \ell_{pd}\}.\ P_a}  \parallel  \underbrace{\Sigma_{a}\{v'/x\}}_{\Sigma'_a} \parallel B_{a}\ ] \nfreccia 
\\
& \ell_{cp} [ B'_{cp} ] \mid 
\ell_a [\ P_a  \parallel  \Sigma'_a \parallel B_{a} \parallel \OUTM{car, v'}{\{\ell_s, \ell_{pd}\}.0} ] 
\end{align}
where we apply the rule (CongrB) to the process $P_a$ for the step from (3.6) to (3.7); 
the rule (Multi-com) is used to move from (3.7) to (3.8), where $v'$ is the value sent by the node $\ell_a$;
since in the message $\OUTM{v'}{\emptyset}. \ 0$ the receiver set is empty, the rule (CongrB) drives the transition from (3.8) to (3.9);
the last transition from (3.9) to (3.10) is performed by applying the rule (Ev-out) where $v'$ is forwarded to nodes $\ell_s$ and $\ell_{pd}$.
Of course the transitions of part of the nodes are preserved by the whole system of nodes $N$, by applying the (ParN) rule.

%% file: sections/analysis.tex
% !TEX root = ../main_coord2016.tex

This section introduces a static analysis for \IoTLySa\ that safely approximates the abstract behaviour of a system of nodes $N$, regardless of the environment that hosts it.
Our analysis tracks the usage of sensor values inside the local node where they are gathered and their propagation in the network of nodes both as row data or processed via suitable functions.
It also describes which messages a node can receive and from which nodes, so abstractly representing the communication structure of the system.

Technically, we define a Control Flow Analysis (CFA, for short), specified in terms of Flow Logic~\cite{NiNi02}, a declarative approach borrowing from and integrating many classical static techniques~\cite{CousotCousot77-1,Heintze:1994,piercetpl,Khedker:2009}.
The distinctive feature of Flow Logic is to separate the \emph{specification} of the analysis from its actual \emph{computation}.
Intuitively, the specification describes when its results, namely analysis \emph{estimates}, are {\em valid}.
Formally, a specification consists of a set of clauses defining the validity of estimates. 
Furthermore, Flow Logic provides us with a methodology to define a correct and efficient analysis algorithm, by reducing the specification to a constraint satisfaction problem.

Below, we specify our analysis in a logical form by inducing on the syntax of the constructs of \IoTLySa\ along the line of~\cite{BBDNN_JCS}. 
An algorithmic version of our analysis can be easily derived from its logical specification.
It suffices to induce on the syntax for generating a set of constraints in AFPL, a logic used to specify static analyses~\cite{NielsonAFPL}.
The generated constraints are solved in low polynomial time obtaining the minimal valid estimate through the succinct solver presented in~\cite{NielsonAFPL}.

%%%\textcolor{green}{
%%%Furthermore, we detect which actions of actuators are actually triggered.
%%%}
%The result might suggest to use a simpler actuator if some of its actions are never exercised, or even to remove it if it is never used.

%The analysis involves a new abstraction for data generated by an IoT system w.r.t.~that in~\cite{BDFG_Coord16,BDFG_ICE2016}.
%Being possibly recursive, these data can be of unbounded size. 
%Therefore, we represent them as trees.

\subsection{Abstract representation of data}
In the following we represent data generated by an IoT system as trees.
The leaves of such trees either identify the sensor from which an observable comes or represent a basic value, and the nodes represent the aggregation functions applied to data.
Since a system is designed to be continuously active and may contain feedback loops, aggregated data can grow unbounded, and so are the trees used to abstractly represent them.
In order to have a finite representation we shall rely here on regular tree grammars~\cite{tata2007}, rather than coalescing all trees exceeding a given depth in a special abstract value as in~\cite{BDFG_Coord16,BDFG_ICE2016}.
These new abstractions better approximate actual data, and the results of the analysis are accordingly more precise.
\medskip

\noindent
A {\em regular tree grammar} is a quadruple $\widehat{G} = ({\mathbb N}, {\mathbb T}, Z, R)$ where
\begin{itemize}
\item  ${\mathbb N}$ is a set of non-terminals (with rank 0),
 \item  ${\mathbb T}$ is a ranked alphabet, whose symbols have an associated arity,
 \item $Z \in {\mathbb N}$ is the starting non-terminal,
\item $R$ is a set of productions of the form $A \rightarrow t$, where $t$ is a tree composed from symbols in ${\mathbb N} \cup {\mathbb T}$ according to their arities.
\end{itemize}
\noindent
In the following we denote the language generated by a given grammar $\widehat{G}$ with $Lang(\widehat{G})$.

%It is convenient introducing some notation first.
Given a system of nodes $N$, the grammars we use will have the alphabet $\mathbb T$ consisting of the following set of ranked symbols 
\begin{itemize}
\item $i^{\ell}$ (with arity 0) for each sensor  $i \in {\mathcal I}_\ell$
\item $v^{\ell}$ (with arity 0) for each node $\ell \in {\mathcal L}$
\item $f^\ell$ (with arity r)  for each function  $f \in {\mathcal F} \text{ and } \ell \in {\mathcal L}$
\end{itemize}
The non-terminals $\mathbb N$ of our grammars include a symbol for each terminal, and just for readability we shall use their capital counterparts, i.e.\ $I^\ell, V^\ell$ and $F^\ell$.
In this way, the production $F^\ell \rightarrow f^\ell(t_1, ..., t_r)$ generates the tree rooted in $f^\ell$ and children generated by $t_1, ..., t_r$.

It is convenient introducing some notation.
For brevity and when no ambiguity may arise, we will simply write $\hat{v} = (Z, R)$ for the grammar $\widehat{G} = ({\mathbb N}, {\mathbb T}, Z, R)$ with starting non-terminal $Z$ and regular productions in $R$, without explicitly listing the terminals and the non-terminals.
Then, we denote with $\mathbb R$ the set of all possible productions over ${\mathbb N}$ and ${\mathbb T}$.

As an example of a possible infinite abstract tree, consider two nodes $N_{\ell_0}$ and $N_{\ell_1}$ with two aggregation functions $f$ and $h$.
Suppose that $N_{\ell_0}$ applies $f$ to a value read by the sensor $i_0$ and a value received from $N_{\ell_1}$.
Similarly, $N_{\ell_1}$ applies $h$ to a value read by the sensor $i_1$ and the value received from $N_{\ell_0}$.
The resulting value is abstracted as the possibly infinite binary tree in Figure~\ref{fig:abst-tree} that belongs to the language of the following grammar:
\[
(F^{\ell_0}, \{
F^{\ell_0} \rightarrow f^{\ell_0}(I^{\ell_0}_0, H^{\ell_1}),
I^{\ell_0}_0 \rightarrow i^{\ell_0}_0,
H^{\ell_1} \rightarrow h^{\ell_1}(I^{\ell_1}_1, F^{\ell_0}),
I^{\ell_1}_1 \rightarrow i^{\ell_1}_1
                  \}
)
\]

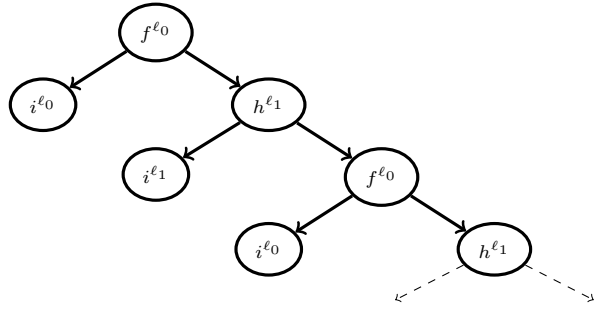
\begin{figure}
\begin{minipage}[c]{0.4\textwidth}
\begin{tikzpicture}[every node/.style={very thick,ellipse},align=center, node distance=0.2cm,auto,scale=0.1,font=\sc\tiny]

\node[draw]                              		 (C1)       {$f^{\ell_0}$};

\node[draw,below=of C1, xshift=-1.5cm]      (C2)       {$i^{\ell_0}$};

\node[draw,below=of C1, xshift=1.5cm]       (C3)       {$h^{\ell_1}$};

\node[draw,below=of C3, xshift=-1.5cm]                  		(C4)       {$i^{\ell_1}$};
\node[draw,below=of C3, xshift=1.5cm]                  		(C5)       {$f^{\ell_0}$};

 \draw[->,very thick]   (C1) -- node[above, xshift = -0.3cm, yshift=-0.3cm] {} (C2);

\draw[->,very thick]   (C1) -- node[above, xshift=0.3cm,yshift=-0.3cm] {} (C3);

 \draw[->,very thick]   (C3) -- node[above, xshift = -0.3cm, yshift=-0.3cm] {} (C4);

\draw[->,very thick]   (C3) -- node[above, xshift=0.3cm,yshift=-0.3cm] {} (C5);

\node[draw,below=of C5, xshift=-1.5cm]      (C6)       {$i^{\ell_0}$};
\node[draw,below=of C5, xshift=1.5cm]       (C7)       {$h^{\ell_1}$};

\draw[->,very thick]   (C5) -- node[above, xshift = -0.3cm, yshift=-0.3cm] {} (C6);

\draw[->,very thick]   (C5) -- node[above, xshift=0.3cm,yshift=-0.3cm] {} (C7);

\node[below=of C7, xshift=-1.5cm]                  		(C8)       {};
\node[ below=of C7, xshift=1.5cm]                  		(C9)       {};

\draw[->,dashed]   (C7) -- node[above, xshift = -0.3cm, yshift=-0.3cm] {} (C8);

\draw[->,dashed]   (C7) -- node[above, xshift=0.3cm,yshift=-0.3cm] {} (C9);

\end{tikzpicture}
\end{minipage}
\caption{An infinite abstract tree}
\label{fig:abst-tree}
\end{figure}

Now we are ready to introduce the \emph{abstract terms} that belong to the set
\[
\widehat{\mathcal V} = 2^{{\mathbb N} \times {\mathbb R}}
\]

\subsection{Specification of the analysis}\label{subsect:analysis}

%%%%%%%%%%
%\paragraph{Terms.}
%%%%%%%%%%
The result of our CFA is a pair $(\widehat{\Sigma}, \Theta)$ for terms $E$ and a 
%%%quadruple $(\widehat{\Sigma},\kappa,\Theta,\alpha)$, 
triple $(\widehat{\Sigma},\kappa,\Theta)$ for $N$,
called \emph{estimate} for $E$ and for $N$, respectively.
The components of an estimate are the following {\em abstract domains}:
\begin{itemize}
\item
{\it abstract store}
$
\widehat{\Sigma} = \bigcup_{\ell \in \mathcal L}\,\widehat{\Sigma}{_\ell}: {\mathcal X} \cup {\mathcal I}_\ell \ \rightarrow\  2^{\widehat{\mathcal V}}
$
where each {\em abstract local store} $\widehat{\Sigma}_{\ell}$ approximates the concrete local store $\Sigma_{\ell}$.
It associates with each location a set of abstract values representing the possible concrete values that the location may store at run time.
\item {\it abstract network environment} 
$\kappa : {\mathcal L} \ \rightarrow\  {\mathcal L} \times \bigcup_{i = 1}^m \widehat{\mathcal V}^i$ where $m$ is the maximum arity of the messages exchanged in the system under analysis (with
$\widehat{\mathcal V}^{i+1} = \widehat{\mathcal V} \times \widehat{\mathcal V}^i$).
Intuitively, $\kappa$ includes all the abstract messages that may be received by the node labelled $\ell$.
\item
{\em abstract data collection} $\Theta: {\mathcal L} \ \rightarrow\ 2^{\widehat{\mathcal V}}$
that, for each node labelled $\ell$, approximates the set of abstract values that the node handles.
%%\item
%%\textcolor{green}{
%%{\em abstract node action collection} $\alpha: {\mathcal L} \times {\mathcal J} \rightarrow Act$ that includes all the  actions $\gamma$ of the
%%actuator $j$ that may be triggered by a control process in the node $\ell$. 
%%For simplicity, we will write $\alpha_{_{\ell}}(j)$ in place of $\alpha(\ell,j)$.
%%}
\end{itemize}

\noindent
The syntax directed rules of \tablename s~\ref{analysisT} and~\ref{analysis} specify when an analysis estimate is valid and they are almost in AFPL format~\cite{NielsonAFPL}. 
For each term $E$, the judgement $\nEFORM{E}{\vartheta}$ expresses that $\vartheta \in  2^{\widehat{\mathcal V}}$ approximates the set of values that $E$ may evaluate to, given the component $\widehat{\Sigma}_{\ell}$ of the abstract store $\widehat{\Sigma}$.
A sensor identifier and a value evaluate to the set $\vartheta$, provided that their abstract representations belong to $\vartheta$ (rules (E-sen) and (E-val)).
This abstract representation is a grammar made of a non-terminal symbol whose production generates a tree with a single node.
For example, the abstract value for a sensor $i$ is $(I^\ell, \{I^\ell \rightarrow i^{\ell}\})$ that represents a grammar with the initial symbol is $I^\ell$ that only generates the tree  $i^\ell$.
Similarly in rule (E-var) where a variable $x$ evaluates to $\vartheta$, if this includes the set of values bound to $x$ in $\widehat{\Sigma_{\ell}}$.
The rule (E-fun) analyses the application of a $r$-ary function $f$ to produce the set $\vartheta$.
To do that
(i) for each term $E_i$, it finds the sets $\vartheta_i$, and
(ii) for all $r$-tuples of values $(\hat{v}_1,\cdots,\hat{v}_r)$ in $\vartheta_1\times\cdots\times\vartheta_r$, it checks if
$\vartheta$ includes the grammars with distinct symbol $F^\ell$ generating the trees rooted in $f^\ell$ with subtrees $\hat{v}_1,\cdots,\hat{v}_r$.
Also, in all the rules for terms, we require that the abstract data collection $\Theta(\ell)$ includes all the abstract values in $\vartheta$.

\begin{table*}[!tb]
\footnotesize
\begin{mathpar}
{
\inferrule[(E-sen)]%
{{(I^\ell, \{I^\ell \rightarrow i^{\ell}\})} \in \vartheta \subseteq  \Theta(\ell)}%
{\nEFORM{i}{\vartheta}}

\inferrule[(E-val)]%
{(V^\ell, \{V^\ell \rightarrow v^{\ell}\}) \in \vartheta \subseteq  \Theta(\ell)}%
{\nEFORM{v}{\vartheta}}

\inferrule[(E-var)]%
{\widehat{\Sigma}{_\ell}({x}) \subseteq \vartheta \subseteq  \Theta(\ell)}%
{\nEFORM{x}{\vartheta}}

}

\inferrule[(E-fun)]%
{ \bigwedge_{i=1}^r\,
   \nEFORM{E_i}{\vartheta_i} \; \wedge  \\ \\
   \forall\, (Z_1, R_1),\cdots, (Z_r, R_r): \ 
   \bigwedge_{i=1}^r\, (Z_i, R_i) \in \vartheta_i\ 
   \Rightarrow
   \ (F^\ell, \{F^\ell \rightarrow f^\ell(Z_1, \cdots, Z_r) \} \, \cup\, \bigcup_{i=1}^r R_i) \in \vartheta \subseteq  \Theta(\ell)
}%
{\nEFORM{f(E_1,\cdots,E_r)}{\vartheta}}

\end{mathpar}

\caption{Analysis of terms $\nEFORM{E}{\vartheta}$.}\label{analysisT}
\end{table*}

In the analysis of nodes we focus on which values can flow through the network and which can be assigned to variables.
%%%\textcolor{green}{
%%%, as well as on the actions processes may trigger over actuators.
%%%}
The judgements have the form $\nNAFORM{N}$
and are defined by the rules in \tablename~\ref{analysis}. 
The rules for the inactive node (N-nil) and for parallel composition (N-par) are standard, as well as 
the rules (B-nil) and (B-par) for node components.
The rule (N-node) for a single node $\ell:[B]$ requires that its component $B$ is analysed, with the further judgment
$\nPAFORM{\ell}{B}$, where $\ell$ is the label of the enclosing node.
The rule (B-store) connects actual stores $\Sigma$ with abstract ones $\widehat{\Sigma}$ and it requires the locations of sensors to contain the corresponding abstract values.
%The rule for sensors is  trivial, because we are only interested in who will use their values, and so is that for actuators.
The rule (B-sen) for sensors does not inspect their form, because we are only interested in who will use their values and this information can be retrieved by the abstract store. 
The same happens in rule (B-act) for actuators that in our model are passive entities.
Indeed, they obey commands of control processes and
to track the activities of actuators it thus suffices to consider the issued commands.
%%%\textcolor{green}{
%%%The same happens in rule (B-act) for actuators that in our model are passive entities.
%%%Indeed, they obey commands of control processes and
%%%to track the activities of actuators it thus suffices to consider the issued commands.
%%%}
%
\begin{table*}
\footnotesize
\begin{mathpar}
{
\inferrule[(N-nil)]%
{ }%
{\nNAFORM{\NIL}}

\inferrule[(N-par)]%
{ \nNAFORM{N_1} \wedge \nNAFORM{N_2}}%
{\nNAFORM{N_1 \ | \ N_2}}

\inferrule[(N-node)]{\nPAFORM{\ell}{B}}{\nNAFORM{\ell:[B]}}

\inferrule[(B-nil)]%
{ }%
{\nPAFORM{\ell}{\NIL}}
}

{
\inferrule[(B-store)]%
{ \forall\, i \in {\mathcal I_\ell}. \ (I^\ell, \{I^\ell \rightarrow i^\ell\}) \in \widehat{\Sigma}_\ell(i)}
{\nPAFORM{\ell}{\Sigma}}

\inferrule[(B-sen)]%
{ }%
{\nPAFORM{\ell}{S}}

\inferrule[(B-act)]%
{ }%
{\nPAFORM{\ell}{A}}

\inferrule[(B-par)]%
{\nPAFORM{\ell}{B_1} \wedge \\\\ \nPAFORM{\ell}{B_2}}%
{\nPAFORM{\ell}{B_1 \| \ B_2}}

}

{
\inferrule[(P-out)]%
{
% \begin{array}{c}
          \bigwedge_{i=1}^{k}\; \nEFORM{E_i}{\vartheta_i} \  \wedge \ 
          \nPAFORM{\ell}{P} \ \wedge 
%          \\[.5ex]
          \forall \hat{v}_1,\cdots,\hat{v}_r:\; \bigwedge_{i=1}^r\, \hat{v}_i \in \vartheta_i\
          \Rightarrow
          \forall \ell' \in L:
         (\ell,  \mess{\hat{v}_1,\cdots,\hat{v}_r}) \in  \kappa(\ell')           
%  \end{array}
}
{\nPAFORM{\ell}{\OUTM{E_1,\cdots,E_r}{L}.\,P}}

}%

{%
\inferrule[(P-in)]%
{ \bigwedge_{i=1}^{j}\;  \nEFORM{E_i}{\vartheta_i}  \ \wedge \ 
          \forall (\ell_2, \mess{ \hat{v} _1,\cdots,\hat{v}_r}) \in \kappa(\ell_1):\;Comp(\ell_2,\ell_1) 
\Rightarrow
\left(          \bigwedge_{i=j+1}^{k}\;  \hat{v}_i \in \widehat{\Sigma}_{\ell_1}({x_i})\ \wedge \
           \nPAFORM{\ell_1}{P} 
\right)}%
{\nPAFORM{\ell_1}{\INPS{E_1,\cdots,E_j}{x_{j+1},\cdots,x_r}{P}}}

}

\inferrule[(P-ass)]%
{\nEFORM{E}{\vartheta} \  \wedge \\\\ \forall \, \hat{v} \in \vartheta \  \Rightarrow
 \hat{v}  \in \widehat{\Sigma}_{_\ell}(x) \ \wedge
\nPAFORM{\ell}{P}}
{\nPAFORM{\ell}{x : = E}.\,{P}}

\inferrule[(P-cond)]%
{\nEFORM{E}{\vartheta} \; \wedge \;  \nPAFORM{\ell}{P_1} \; \wedge \ \nPAFORM{\ell}{P_2}}
{\nPAFORM{\ell}{E?P_1 : P_2}}

%\textcolor{green}{
%\inferrule[(P-act)]%
%{\gamma \in \alpha_\ell(j) \ \wedge \ \nPAFORM{\ell}{P}}
%      {\nPAFORM{\ell}{ \OUTS{j,\gamma}{P}}}
%}

\inferrule[(P-act)]%
{\nPAFORM{\ell}{P}}
      {\nPAFORM{\ell}{ \OUTS{j,\gamma}{P}}}
      
\inferrule[(P-rec)]%
{\nPAFORM{\ell}{P}}%
{{\nPAFORM{\ell}{\mu h. \ P}}}

\inferrule[(P-var)]%
{\nPAFORM{\ell}{P}}%
{{\nPAFORM{\ell}{h}}}

\end{mathpar}

\caption{Analysis of nodes $\nNAFORM{N}$, and of node components
$\nPAFORM{\ell}{B}$.}\label{analysis}
\end{table*}

The rules for processes are in the lower part of \tablename~\ref{analysis}, and all require that an estimate is valid for the immediate sub-processes.
The rule (P-out) for $r$-ary multi-output
(i)  finds the sets $\vartheta_i$, for each term $E_i$; and 
(ii) for all $r$-tuples of values $(\hat{v}_1,\cdots,\hat{v}_r)$ in $\vartheta_1 \times\cdots\times \vartheta_r$, it checks if
they belong to $\kappa(\ell' \in L)$, i.e.\ if these $r$-tuples of values can be received by the nodes with labels in $L$.
In the rule (P-in) for input the terms $E_1, \cdots, E_j$ are 
used for matching values sent on the network: 
this rule (i) checks whether 
 the first $j$ terms have valid estimates $\vartheta_{i}$; and
(ii) for each message $(\ell_2, \mess{\hat{v}_1,\cdots,\hat{v}_j,\hat{v}_{j+1},\ldots,\hat{v}_r})$ in $\kappa(\ell_1)$ 
(i.e.~in any message predicted to be receivable by the node with label $\ell_1$), it checks that 
% are pointwise included in $\vartheta_i$; and
the values $\hat{v}_{j+1},\ldots,\hat{v}_r$ are included in the estimates for the variables $x_{j+1},\cdots,x_r$, provided that the two nodes can communicate ($Comp(\ell_2,\ell_1)$).
%Thus, this rule checks whether (i) 
%these first $j$ terms have acceptable estimates
%$\vartheta_{i}$; (ii) the first $j$ values of any
%message $(\ell', \mess{\hat{v}_1,\cdots,\hat{v}_j,\hat{v}_{j+1},\ldots,\hat{v}_k})$ in $\kappa(\ell)$ 
%(i.e.~in any message predicted to be receivable by the node with label $\ell$) are pointwise
%included in $\vartheta_i$;
%(iii) the two nodes can communicate ($Comp(\ell',\ell)$);
%(iv) the values $\hat{v}_{j+1},\ldots,\hat{v}_k$
%are included in the estimates for the variables $x_{j+1},\cdots,x_k$.
The rule (P-ass) for assignment
requires that $\widehat{\Sigma}_\ell(x)$ includes all the values $\hat{v}$ in $\vartheta$, the estimate for $E$.
The rule (P-cond) is as expected and the rule (P-act) is trivial.
Finally, the rules (P-rec) and (P-var) for iteration are standard, where to save notation, we assumed that each variable $h$ is uniquely bound to the body $P$.

%\subsection{Running example}
To show how our analysis works, consider again the example in Section~\ref{sec:example} and the process 
$P_{cp} = \mu h. (z := 1).(z' := noiseRed(z)).\OUTM{z'}{\{\ell_a\}}. \ h$.
Moreover, let $\iota$ and $\nu$ be 
\begin{align}
& \iota = (\mathcal{I}^{cp}, \{\mathcal{I}^{cp} \rightarrow 1^{cp}\})\label{eq:iota}\\
& \nu = (\textsc{NoiseRed}^{cp}, \{\textsc{NoiseRed}^{cp} \rightarrow noiseRed^{cp}(\mathcal{I}^{cp}), \mathcal{I}^{cp} \rightarrow 1^{cp}\})\label{eq:xi}
\end{align}
abstracting values of the camera $S_{cp}$ and its elaboration, respectively.
Every valid CFA estimate must include at least the following entries:
%
%\[
%\begin{array}{l}
%(a)\, \widehat{\Sigma}_{\ell_{cp}}(z) \supseteq \{1^{\ell_{cp}} \}
%\qquad\qquad\qquad\qquad\ \ 
%(b)\, \widehat{\Sigma}_{\ell_{cp}}(z') \supseteq \{noiseRed^{^{\ell_{cp}}}(1^{\ell_{cp}}),1^{\ell_{cp}} \}
%\\
%(c)\, \Theta(\ell_{cp}) \supseteq \{1^{\ell_{cp}},noiseRed^{^{\ell_{cp}}}(1^{\ell_{cp}}) \} 
%\quad
%(d)\, \kappa(\ell_a) \supseteq \{ (\ell_{cp},\mess{noiseRed^{^{\ell_{cp}}}(1^{\ell_{cp}})})\}
%\end{array}
%\]
\begin{align*}
& (a)\, \widehat{\Sigma}_{\ell_{cp}}(z) \supseteq \{ \iota \}
&& (b)\, \widehat{\Sigma}_{\ell_{cp}}(z') \supseteq \{\iota, \nu \}
\\
& (c)\, \Theta(\ell_{cp}) \supseteq \{\iota, \nu \} 
&& (d)\, \kappa(\ell_a) \supseteq \{ (\ell_{cp},\mess{\nu})\}
\end{align*}
Indeed, all the following checks must succeed:
\begin{itemize}
\item
$\nPAFORM{\ell_{cp}}{\mu h. (z := \underline{1}).(z' := noiseRed(z)).\OUTM{z'}{\{\ell_a\}}.h}$ because

\item $\nPAFORM{\ell_{cp}}{ (z := \underline{1}).(z' := noiseRed(z)).\OUTM{z'}{\{\ell_a\}}}$, that in turn holds
%\item in turn, $\nPFORM{\ell_{cp}}{ (z := 1).(z' := noiseRed(z)).\OUTM{z'}{\{\ell_a\}}}$ requires that 
\item because (i) $\iota$ is in $\widehat{\Sigma}_{\ell_{cp}}(z)$ by ({\it a}) ($\nEFORM{\underline{1}}{\vartheta} \ni \iota$); and because 

%\item because (i) $1^{\ell_{cp}}$ is in $\widehat{\Sigma}_{\ell_{cp}}(z)$ by ({\it a}) ($\nEFORM{1}{\vartheta} \ni 1^{\ell_{cp}}$); and because 

(ii) $\nPAFORM{\ell_{cp}}{ (z' := noiseRed(z)).\OUTM{z'}{\{\ell_a\}}}$, that in turn holds
\item
because (i) $\nu$ is in $\widehat{\Sigma}_{\ell_{cp}}(z')$ by ({\it b}) since
 
 $\nEFORMM{\ell_{cp}}{noiseRed(z)}{\vartheta}$ $\ni \nu$; and because

(ii) $\nPAFORM{\ell_{cp}}{\OUTM{z'}{\{\ell_a\}}}$ that holds because
$(\ell_{cp},\mess{\nu})$ is in $\kappa(\ell_a)$ by ({\it d}).
\end{itemize}

\noindent
The precision of the CFA above can be refined by replacing the abstract store $\widehat{\Sigma}$ with the pair
 $\widehat{\Sigma}_{in}, \widehat{\Sigma}_{out}$.
 This extension requires a more verbose specification of the rules for accurately handling the store updates, similarly to the treatment of side effects in~\cite{NiNi02}.
We can obtain a further improvement of the precision by making the analysis more context-sensitive.
In particular, an additional component can record the sequence of choices made in conditionals while traversing the node under analysis.
One can thus obtain better approximations of the store or detect causal dependencies among the data sent by sensors and the actions carried out by actuators, as well as casuality among nodes.

\subsection{Correctness of the analysis.}
Our CFA respects the operational semantics.
%As usual, we prove a subject reduction result and the existence of a (minimal) estimate.
The proof of this fact benefits from an instrumented denotational semantics for expressions.
For that we will introduce a set of trees $\widehat{\mathcal{T}}$, ranged over by $\hat{t}$, $\hat{t'}$, \dots, built over the ranked alphabet $\mathbb{T}$ introduced above. 
The values of expressions become now pairs $\langle v, \hat{t}\rangle$, and the store and its updates are accordingly extended.
The instrumented local store then becomes
$
\Sigma^{i}_{\ell} : \mathcal{X}_{\ell} \cup \mathcal{I} _\ell\ \rightarrow \mathcal{V} \times \widehat{\mathcal T}.
$ 
We also endow $\Sigma^i_\ell$ with an undefined value $\bot$ for recording when a sensor or a variable is not initialised.
Accordingly, we assume that the semantic function $measure(i,\mathcal{E})$ returns a pair $(v, i^\ell)$.
Finally, when no ambiguity arises, we shall overload the symbol $v$ to also denote the instrumented values.

The formal definition of the instrumented denotational semantics follows, where we indicate with $\downarrow_i$ the projection on the $i^{th}$ component of the pair.
\[
\begin{array}{lll}
\dsem{v}^i_{\Sigma^{i}_{\ell}} & = & (v, v^\ell)
\\
\dsem{i}^i_{\Sigma^{i}_{\ell}} & = & \Sigma_{\ell}^i(i)
\\
\dsem{x}^i_{\Sigma^{i}_{\ell}} & = & \Sigma_{\ell}^i(x)
\\
\dsem{f(E_1,\cdots, E_r)}^i_{\Sigma^{i}_{\ell}} & = & 
(f(\dsem{E_1}^i_{\Sigma^{i}_{\ell} \downarrow_1},\cdots, \dsem{E_r}^i_{\Sigma^{i}_{\ell} \downarrow_1}),
f^{\ell}(\dsem{E_1}^i_{\Sigma^{i}_{\ell} \downarrow_2},\cdots,\dsem{E_r}^i_{\Sigma^{i}_{\ell} \downarrow_2}))
\end{array}
\]
%%%%%%%%% DEF + STRETTA %%%%%%%%%%%%
%\[
%\begin{array}{l}
%\dsem{v}^i_{\Sigma^{i}_{\ell}}  =  (v, v^\ell)
%\qquad\qquad\qquad
%\dsem{i}^i_{\Sigma^{i}_{\ell}}  =  \Sigma_{\ell}^i(i)
%\qquad\qquad\qquad
%\dsem{x}^i_{\Sigma^{i}_{\ell}}  =  \Sigma_{\ell}^i(x)
%\\
%\dsem{f(E_1,\cdots, E_k)}^i_{\Sigma^{i}_{\ell}}  =  
%(f(\dsem{E_1}^i_{\Sigma^{i}_{\ell} \downarrow_1},\cdots, \dsem{E_k}^i_{\Sigma^{i}_{\ell} \downarrow_1}),
% f^{\ell}(\dsem{E_1}^i_{\Sigma^{i}_{\ell} \downarrow_2},\cdots,\dsem{E_k}^i_{\Sigma^{i}_{\ell} \downarrow_2}))
%\end{array}
%\]
%

\noindent
Since the CFA only considers the second component of the extended store, we need to define when the concrete and the abstract stores agree.

\begin{definition}
Given a concrete store $\Sigma_\ell^i$ and an abstract store $ \widehat{\Sigma}_\ell$, we say that
they {\em agree}, in symbols
$\Sigma_\ell^i \bowtie \widehat{\Sigma}_\ell$, if and only if for all
$w \in \mathcal{X}_{\ell} \cup \mathcal{I} _\ell$ either $\Sigma^i_\ell (w) = \bot$
or there exists $\widehat{G} \in \widehat{\Sigma}_\ell(w)$ such that 
$(\Sigma^i_\ell (w))_{\downarrow_2} \in  Lang(\widehat{G})$.
\end{definition}
\noindent
Just to give an intuition, suppose that the expression $E$ is such that $\dsem{E}_{\Sigma^i_\ell}^i  =  (v, v^\ell)$. 
Then, the assignment $x : = E$ will result in the updated store $\Sigma^i_\ell \{(v, v^\ell)/x\} $.
Clearly, the standard semantics of expressions used in \tablename~\ref{opsem} is obtained by the projection on the first component of the instrumented one.
In our running example, the assignment $z' := noiseRed(z)$ of the process $P_{cp}$ stores the pair $(v, noiseRed^{\ell_{cp}}(1^{\ell_{cp}}))$ made of the actual value $v$ and of its abstract counterpart.

The following theorem establishes the correctness of our CFA in that its valid estimates are preserved under reduction steps.

\begin{restatable}[Subject reduction]{thm}{subjectreduction}\label{SRtheorem}
Given a system of nodes $N$, if 
$\nNAFORM{N}$; $N \sfreccia N'$ and $\,\forall\,\Sigma_\ell^i$ in $N$ it is $\Sigma_\ell^i \bowtie \widehat{\Sigma}_\ell$, 
then $\nNAFORM{N'}$ and $\,\forall\,  \Sigma{_\ell^i}'$ in $N'$ it is $\Sigma{_\ell^i}' \bowtie \widehat{\Sigma}_\ell$.
\end{restatable}

We also prove that the set of valid estimates to the specification in \tablename~\ref{analysis} is never empty and that a minimal one always exists.
This is because estimates form a Moore family ${\mathcal M}$, i.e.\ a set with a greatest element ($\sqcup \emptyset$) and a least element ($\sqcup {\mathcal M}$). 

\begin{restatable}[Existence of estimates]{thm}{extsolutions}\label{Mooretheorem}
Given $N$, its valid estimates form a Moore family.
\end{restatable}

%The following corollary of subject reduction explains the role that the components of a valid estimate play in predicting the behaviour of the analysed system.
We now illustrate the role that the components of a valid estimate play in predicting the behaviour of the analysed system.
The following corollary follows from the fact that the analysis respects the operational semantics, as stated in Theorem~\ref{SRtheorem}.
The first item below makes it evident that our analysis determines whether the value of a term may indeed be used along the computations of a system, and clarifies the role of the component $\Theta$;
the second item guarantees that $\kappa$ predicts all the possible inter-node communications.
%%%\textcolor{green}{
%%%; the third item says that the component $\alpha$ safely predicts which actions of actuators are triggered by processes.
%%%}
In the statement of this corollary we use the following notation for reductions, where we omit writing the environments for readability.
Let $N \xrightarrow{E_1, ...,E_r}_\ell N'$ denote a reduction in which all the expressions $E_i$ are evaluated at node $\ell$ and let
$N \xrightarrow{\mess{v_1,\dots,v_r}}_{\ell_{1}, \ell_{2}} N'$ be a reduction in which the message sent by node $\ell_1$ is received by node $\ell_2$.
%%%\textcolor{green}{
%%%; and let $N \xrightarrow{\langle{j,\gamma}\rangle}_{\ell} N'$ be a reduction where a process of the node $\ell$ commands the actuator $j$ to perform the action $\gamma$. 
%%%}

\begin{restatable}{coro}{corTheta}\label{cor:Theta}
%Let $N \xrightarrow{E_1, ...,E_n}_\ell N'$ denote a reduction in which all $E_i$ are evaluated at node $\ell$.
%If $\nNFORM{N}$ and $N \xrightarrow{E_1, ...,E_n}_\ell  N'$ then $\forall k \in [0,n]$ 
%there exists $\hat{v} \in \Theta(\ell)$ sub-term of $(\dsem{E_i}^i_{\Sigma^i_\ell})_{\downarrow_2}$.
Assume that $\nNAFORM{N}$, then
\begin{enumerate}
\item
if $N \xrightarrow{E_1, ...,E_r}_\ell  N'$ then $\forall k \in [0,n]$ 
there exists $\widehat{G} \in \Theta(\ell)$ such that \\
$(\dsem{E}^i_{\Sigma^i_\ell})_{\downarrow_2} \in Lang(\widehat{G})$;
\item
if $N \xrightarrow{\mess{v_1,\dots,v_r}}_{\ell_1,\ell_2} N'$ then
it holds $(\ell_{1},\mess{\hat{v}_1,\dots,\hat{v}_r}) \in \kappa(\ell_{2})$, 
and $\forall i \in [0,r]$ there exists $\widehat{G} \in \hat{v}_i$ such that
$v_{i\downarrow_2} \in Lang(\widehat{G})$.
%\item
%if $N \xrightarrow{\langle{j,\gamma}\rangle}_{\ell} N'$ then $\gamma \in \alpha_{\ell}(j)$.
\end{enumerate}
\end{restatable}

%\begin{restatable}{corollary}{corKappa}\label{cor:kappa}\ \\
%%Let $N \xrightarrow{\mess{v_1,\dots,v_n}}_{\ell, \ell'} N'$ denote a reduction in which the message sent by node $\ell$ is received by node $\ell' $.
%%%\\
%%If $\nNFORM{N}$ and $N \xrightarrow{\mess{v_1,\dots,v_n}}_{\ell,\ell'} N'$ then
%%it holds $(\ell,\mess{\hat{v}_1,\dots,\hat{v}_n}) \in \kappa(\ell')$, 
%%where $\hat{v_i}$ is a sub-term of $v_{i\downarrow_2}$.
%Let $N \xrightarrow{\mess{v_1,\dots,v_n}}_{\ell_{1}, \ell_{2}} N'$ denote a reduction in which the message sent by node $\ell_1$ is received by node $\ell_2$.
%%\\
%If $\nNFORM{N}$ and $N \xrightarrow{\mess{v_1,\dots,v_n}}_{\ell_1,\ell_2} N'$ then
%it holds $(\ell_{1},\mess{\hat{v}_1,\dots,\hat{v}_n}) \in \kappa(\ell_{2})$, 
%where $\hat{v_i} = v_{i\downarrow_2}$.
%\end{restatable}
\noindent
Back again to our example at the end of sub-section~\ref{subsect:analysis},
 consider the assignment $z := 1$ in the process $P_{cp}$.
We have $(\dsem{1}^1_{\Sigma^1_{\ell_{cp}}})_{\downarrow_2} = 1^{\ell_{cp}}$
%%$P'_{cp} = (z' := noiseRed(v)).\OUTM{z'}{\{\ell_a\}}.P''_{cp}$, 
where $v$ is the actual value received by the first sensor, and 
that our analysis computes $\iota \in \theta(\ell_{cp})$, where $\iota$ is defined in~(\ref{eq:iota}). 
It is immediate to see that $1^{\ell_{cp}} \in Lang(\iota)$.
A similar reasoning can be done also for messages from $\ell_{cp}$ to $\ell_{a}$. 
Indeed, we have $(\dsem{z'}^1_{\Sigma^1_{\ell_{cp}}})_{\downarrow_2} = noiseRed^{\ell_{cp}}(1^{\ell_{cp}})$ and that our analysis computes 
$(\ell_{cp},\mess{\nu}) \in \kappa(\ell_{a})$, where $\nu$ is defined in~(\ref{eq:xi}).
It is immediate to see that $noiseRed^{\ell_{cp}}(1^{\ell_{cp}}) \in Lang(\nu)$.
%Back again to our example, we have that 
%%\[N_{cp} = \ell_{cp} : [ P_{cp} \ \|\ S_{cp}  \ \|\ B_{cp} ] \xrightarrow{1}_{\ell_{cp}} N'_{cp} = \ell_{cp} : [ P'_{cp} \ \|\ S_{cp}  \ \|\ B_{cp} ]\]
%%Since $\nNFORM{N_{cp}}$, we have that
%$1^{\ell_{cp}} \in \Theta(\ell_{cp})$, where $(\dsem{1}^1_{\Sigma^1_{\ell_{cp}}})_{\downarrow_2} = 1^{\ell_{cp}}$,  
%%$P'_{cp} = (z' := noiseRed(v)).\OUTM{z'}{\{\ell_a\}}.P''_{cp}$, 
%and where $v$ is the actual value received by the first sensor.
%Similarly, we have that
%%\[
%%\begin{array}{c}
%%\ell_{cp} : [ \OUTM{z'}{\{\ell_a\}}.P''_{cp} \ \|\ S_{cp}  \ \|\ B_{cp} ] \mid 
%%\ell_a : [\ \mu h. \INPS{}{x}{\OUTM{car, x}{\{\ell_s\}.\ h}} \ \|\ B_{a}\ ]
%%\\
%%\xrightarrow{v}_{\ell_{cp},\ell_a} \\
%%\ell_{cp} : [ P''_{cp} \ \|\ S_{cp}  \ \|\ B_{cp} ] \mid 
%%\ell_a : [\OUTM{car, v}\{\ell_s\}.P'_{a} \ \|\ B_{a}\ ]
%%\end{array}
%%\]
%%Since $\nNFORM{N_{cp} \mid N_a}$, we have that 
%$(\ell_{cp},\mess{\hat{v}}) \in \kappa(\ell_{a})$, 
%where $\hat{v} = v_{\downarrow_2}$.

\subsection{Some applications of the analysis }

To give an idea on how the outcome of the analysis can be used to detect where and how data are manipulated and how messages flow in a system, we introduce below a couple of simple properties.

A designer could be interested in checking whether a value taken by a specific sensor of a node will be used as an ingredient of the data of a different node.
This is formalised as follows.

\begin{definition}
Let $N$ be the node $\ell:[B]$ containing the sensor $i$.
\\
The sensor $i$ is an {\em ingredient} of a node $N'$ with label $\ell'$ 
if and only if $N' \xrightarrow{E_1, ...,E_r}_{\ell '} N''$ and there exists $k \in [1,r]$ such that 
$(\dsem{E}^i_{\Sigma^i_\ell})_{\downarrow_2}$ is a tree with a leaf $i^\ell$.
\end{definition}

The following property is an immediate consequence of Corollary~\ref{cor:Theta}.

\begin{restatable}[]{prope}{ingredient}\label{Ingredienttheorem}
Let $i$ be a sensor of the node $N$ with label $\ell$; let $N'$ be a node with label $\ell'$; 
and let $\nNAFORM{N \mid N'}$.
\\
The sensor $i$ is an ingredient of $N'$ if there exists $\widehat{G} \in \Theta(\ell')$ and a tree in its language with a leaf $i^\ell$.
\end{restatable}

Consider again our example at the end of sub-section~\ref{subsect:analysis}, for which estimates are such that
 $\Theta(\ell_{a}) \supseteq \{ \nu \}$ (because $ \kappa(\ell_a) \supseteq \{ (\ell_{cp},\mess{\nu})\}$), showing that the sensor $1^{cp}$ is an ingredient of the node $N_a$.

Note that the schema above can be used to verify other properties by inspecting the component $\Theta$.
It suffices to express properties of tree languages and check them with the standard tools, e.g.\ tree automata model checking.

Our next example works over the problem of checking the robustness of the communications in the smart street light example discussed at the end of Section~\ref{sec:example}, where we assumed the lamp post $\ell_q$ to be our of order.

Consider the following abstract values, each for every $p$: 
\[
 \nu^c_p = (\textsc{Car}^{p}, \{\textsc{Car}^{p} \rightarrow car^{p}\}) \qquad \qquad
 \nu^w_p = (\textsc{True}^{p}, \{\textsc{True}^{p} \rightarrow true^{p}\}) 
\]
the first is used to signal the presence of a car, the second one represents a boolean for the pedestrian. 
Every valid CFA estimate for our example includes in the component $\kappa$ the following entries for all lamp posts $p$:
\[
\kappa(\ell_p) \supseteq \{ (\ell_{p-1}, \mess{\nu^c_{p-1}}),\, (\ell_{p-1}, \mess{\nu^w_{p-1}}),\, (\ell_{p+1}, \mess{\nu^w_{p+1}}) \}
\]
In particular, this holds also for the lamp post $\ell_{q+1}$ and the designer can check that the messages $(\ell_{q}, \mess{\nu^c_{q}}),\, (\ell_{q}, \mess{\nu^w_{q}})$ are never received because the lamp post $\ell_q$ is out of order.

If needed, the designer can make this situation more explicit, by encoding the fault of $\ell_q$ in the compatibility function.
To this aim, it suffices to modify the function only for the faulty lamp post, by putting $Comp(\ell_q,\ell_p) = \mathit{false}$ for all lamp posts.
Then, the designer can re-apply the analysis and obtain that $(\ell_{q}, \mess{\nu^c_{q}}),\, (\ell_{q}, \mess{\nu^w_{q}}) \notin \kappa(\ell_{q+1})$, from which deducing
that these messages will never be received by the lamp posts from position $q+1$ onwards.

%% file: sections/extensions.tex
% !TEX root = ../main_coord2016.tex

Security is a crucial issue in IoT since it ``deals not only with huge amount of sensitive data (personal data, business
data, etc.) but also has the power of influencing the physical environment
with its control abilities''~\cite{IERC}.
Traditional security risks are magnified~\cite{Bertino1} in the IoT world, where opportunities for potential breaches dramatically increase hand in hand with the number of connected {\em things}.
Actually, security concerns both the physical and the logical facets of cyber-physical systems and poses manifold challenges.
%Actually, it concerns both the physical and the logical facets of cyber-physical systems.
However, these aspects have not received much attention so far.
Counter-measures against a physical attack are very hard to implement, 
sensors are often in not protected locations and can thus easily be damaged and their transmissions disrupted.
Yet more standard, security of the logical components raises new issues, mainly because IoT systems must perform all in a lightweight way.
%\marginpar{\alert{Si vuole citare bertino1.pdf?}}

%Their investigation are still in its infancy.
%\textcolor{cyan}{Dealing with security attacks is harder than ever, because IoT systems must perform all 
%in a lightweight way due to 
%the low capacities of the involved devices.}
Here we extend our proposal with encryption and decryption constructs for representing and handling some aspects of security from the side of applications.
Encryption is indeed the traditional mechanism to provide security, but it is very expensive in a setting where energy, computation and storage are rather limited.
Our CFA may help IoT designers in singling out the pieces of information to be kept confidential and in detecting whether the security assets are protected inside the system. 
In addition, we will show that a static check suffices to detect if certain data are propagated within a system of nodes according to a given policy.
In the first part of this section we propose a general schema that captures a family of security and access control policies.
We then instantiate the schema to a couple of classical security notions and a more sophisticated policy that controls how information is propagated.

%Therefore designers can carefully choose when and where resorting to costly encryptions.
%
%\textcolor{cyan}{
%In addition they can evaluate where cryptography is necessary or redundant.
%Indeed by classifying and partitioning (abstract) values as confidential or
%public, we can 
%%check whether values that are classified in a node are never sent to other untrusted nodes and therefore we can 
%detect possible leakages.}
%Here we extend our proposal with encryption and decryption constructs so as to represent and handle some aspects of security from the side of applications.

%In particular, we show how our CFA can be used to detect leakages, once values and data are partitioned in confidential and public.
In particular, we show how our CFA can address data security, by detecting leakages, once values and data are partitioned in confidential and public.
The designer of a system of nodes can then suitably protect sensitive information, by only encrypting their relevant parts, typically the values coming from sensors or from data aggregations.
In addition, we re-cast the classical no-read-up/no-write down in our framework.
Indeed, we can then statically check if information is exchanged between nodes in the right direction with respect to their level of clearance.
%given clearance levels of security.
Finally, we instantiate the general schema with a finer policy that restricts the propagation of certain confidential data only within a specific sub-system.
We show that our CFA suffices to detect whether some confidential information unintentionally reaches a node that should not get it.
Remarkably, verifying the above properties, and others, is done on the estimates of our CFA that are computed once and for all.

We are here assuming a symmetric encryption schema, which is the one commonly used in IoT systems because less energy consuming; however, using public key  encryption will only require slight changes in the definitions below.
Formally, the encryption of a message $m$ under a key $k$ can be defined as the result of an aggregation function $encrypt(m,k)$ and the decryption as the application of its inverse $encrypt^{-1}(x,k)$.
%(we are assuming a symmetric encryption schema).
The machinery developed so far would then formally work as it is. 
However the importance of security calls for an explicit treatment of its aspects, described in this section.

Our current analysis is restricted to the system of nodes under investigation, and does not consider the possible presence of active malicious attackers that can intervene in communications by injecting forged messages.
We do not feel these restrictions too heavy.
One reason is that in many real IoT systems, e.g.\ those developed with Zigbee~\cite{ZigBee}, keys are seldom exchanged in communications.
In this real world systems cryptographic keys, if present at all, are often fixed once and for all and exchanged in a secure manner when the system starts its setup or upon a reboot.
Also, it is not difficult to handle the case when the design of the system benefits from a key distribution.
Indeed, one can re-use existent techniques for doing that, e.g.\ based on CFA~\cite{BBDNN_JCS, NielsonZigB}, through which 
also active attackers are implicitly included.

%We do not feel this to be a heavy limitation, because often in real IoT systems (see e.g.~\cite{ZigBee}) cryptographic keys, if present at all, are fixed once and for all, and they are exchanged in a secure manner when the system starts its setup or upon a reboot.
%Also, keys are seldom exchanged in communications afterwards, because conventional key distribution and key management mechanisms deeply suffer from the lack of computational resources.
%In any case, it would not be difficult to implicitly include attackers, following~\cite{BBDNN_JCS}, or to apply the technique of~\cite{NielsonZigB} if necessary.

\subsection{Extending the language and the analysis}
We extend here the syntax of processes and terms with encryption and decryption primitives as in {\LySa}, taking for simplicity a symmetric encryption schema.
As argued above, we feel free to assume as given a finite set $\mathcal K$ of secret keys owned by nodes, previously exchanged in a secure way.

Terms now include also encryptions written
\[
\{E_1,\cdots,E_r\}_{k}
\]
the result of which comes from encrypting the values of the expressions $E_i$, for $i \in [1,r]$, under the shared key $k$ in $\mathcal K\!$.

The syntax of processes now has also the following construct
\[
 \DECSO{E}{E_1,\cdots,E_j}
      {x_{j+1},\cdots,x_r}{k}{}{P}
\]
This process receives a message encrypted with the shared key $k \in \mathcal K\!$.
Also in this case we use pattern matching, as explained below.
%but additionally 
%the message $E =\{E'_1, \cdots, E'_r\}_{k}$ is decrypted with the key $k$. 
%Hence, whenever $E_i = E'_i$ for all $i\in [0,j]$, the receiving process behaves as $P\{E_{j+1}/x_{j+1},\ldots,E_{k}/x_{k}\}$.

The new decryption construct requires the following semantic rule:
\medskip
\[
\inferrule[(Decrypt)]%
{ \dsem{E}_{\Sigma} = \{v_1,\cdots,v_r\}_{k} \ \wedge \
  \bigwedge_{i=1}^{j} {v_i}= \dsem{E'_i}_{\Sigma}}
{\Sigma \ \| \ \DECSO{E}{E'_1,\cdots,E'_j}{x_{j+1},\cdots,x_r}{k}{}{P} \| \ B
\; \nfreccia \; 
\Sigma\{ v_{j+1}/x_{j+1},\cdots,v_r/x_r\} \ \| \  P \ \| \ B }
\]

In this inference rule the encryption $ \{v_1,\cdots, v_r\}_{k}$, resulting from $E$, has to match against the pattern in $\mathsf{decrypt}\;{E}\;\mathsf{as}\; \{E'_1,\cdots,E'_j ; x_{j+1},\cdots,x_r\}_{k}$
$\mathsf{in}\;{P}$, i.e.\ the first $j$ values $v_i$ must match those of the corresponding ${E'_i}$.
In addition the keys must be the same (this models \emph{perfect} symmetric cryptography). 
When all the above holds, the values of the remaining expressions are bound to the corresponding variables.

In order to extend the analysis for each node $\ell$, we first add to the alphabet $\mathbb{T}$ the following new symbols, with elements:
\begin{itemize}
\item 
$enc^\ell_r$, each with rank $r$ ranging from 2 to $m$, where $m-1$ is the maximum arity of the encrypted terms used in the system;
\item
$k_i$, for each key $k_i \in \mathcal{K}$.
\end{itemize}

\noindent
As done above, we assume that there is a non-terminal for each of the new terminals, written with capital letters.
It is convenient to define the following auxiliary function \texttt{D} for extracting the ordered list of the grammars concerning the abstract  sub-terms of an encryption under the key $k$ (otherwise it returns the empty list):
\[
\mathtt{D}((\textsc{Enc}^\ell_r, R), k) =
[
(Z_i, R_i) \mid \exists\; \textsc{Enc}^\ell_r \rightarrow enc^\ell_r(Z_1, \cdots, Z_r, K), K \rightarrow k \in R \; \land \;
R_i = \mathtt{D}^*(Z, R)
]
\]
where $\mathtt{D}^*(Z, R)$ extracts the productions in $R$ that are used to derive trees from $Z$ and is defined below
\[
\mathtt{D}^*(Z, R) = \{   Z \rightarrow \hat{t} \in R \mid \hat{t} \in \mathbb{T}   \} \; \cup \;
\bigcup_{Z \rightarrow M \in R} \{ \mathtt{D}^*
(W, R) \mid  W \text{ occurs in } M \in \mathbb{N \cup T}
                          \} 
\]

\noindent
To better understand how the function $\mathtt{D}$ works, consider the node $N'_a$ of the example in Section~\ref{sec:example}.
The analysis tells us that the variable $\mathtt{y}$ could take the abstract value $\epsilon$ defined as follows
\[
\begin{array}{ll}
(\textsc{Enc}_1^{cp}, 
 \{ &\hspace{-3mm}
            \textsc{Enc}_1^{cp} \rightarrow enc^{cp}_1(\textsc{NoiseRed}^{cp}, K), K \rightarrow k,\, 
\\ &\hspace{-3mm} \textsc{NoiseRed}^{cp} \rightarrow NoiseRed^{cp}(\mathcal{I}^{cp}), \mathcal{I}^{cp} \rightarrow 1^{cp} \;\}).
\end{array}
\]
The function $\mathtt{D}$ applied to $\epsilon$ extracts the grammar approximating the enhanced version of the picture taken by the sensor $1$ in the node $N_{cp}$:
\[
(\textsc{NoiseRed}^{cp}, \{ \textsc{NoiseRed}^{cp} \rightarrow noiseRed^{cp}(\mathcal{I}^{cp}), \mathcal{I}^{cp} \rightarrow 1^{cp}\}).
\]

We now introduce the following new rules for analysing encrypted terms and decryption processes.

\begin{mathpar}

\inferrule[(E-enc)]%
{ \bigwedge_{i=1}^r\,
   \nEFORM{E_i}{\vartheta_i} \; \wedge  
\ 
  \forall\, (Z_1, R_1),\cdots, (Z_r, R_r): \ 
     \bigwedge_{i=1}^r\, (Z_i, R_i) \in \vartheta_i\ 
   \Rightarrow
\\ \\
   \ (\textsc{Enc}^\ell_r, \{\textsc{Enc}^\ell_r \rightarrow enc^\ell_r(Z_1, \cdots, Z_r, K),   K \rightarrow k\} \, \cup\, 
         \bigcup_{i=1}^r R_i) \in \vartheta \subseteq  \Theta(\ell)
}%
{\nEFORM{\{E_1,\cdots,E_r\}_{k}}{\vartheta}}

\end{mathpar}
\begin{mathpar}

\inferrule[(P-dec)]%
{
 \nEFORM{E}{\vartheta}  \ \wedge \ 
  \bigwedge_{i=1}^j\,
   \nEFORM{E_i}{\vartheta_i} \; \wedge \\ \\
   \forall \; \hat{v} \in \vartheta \text{ s.t. } 
   \texttt{D}(\hat{v}, k) = [\hat{v}_1,\cdots, \hat{v}_r]
     \Rightarrow
\left(         
 \bigwedge_{i=j+1}^r\;  \hat{v}_i \in \widehat{\Sigma}{_\ell}({x_i})\ \wedge \ 
           \nPAFORM{\ell}{P} 
\right)}%
{\nPAFORM{\ell}{{\DECSO{E}{E_1,\cdots,E_j}{x_{j+1},\cdots,x_r}{k}{}{P}}}}

\end{mathpar}

Similarly to the rule for evaluating function applications, (E-enc) analyses an encryption term and produces the set $\vartheta$.
To do that
(i) for each term $E_i$, it finds the sets $\vartheta_i$, and
(ii) for all $r$-tuples of values 
$(\hat{v}_1,\cdots,\hat{v}_r)$ in $\vartheta_1 \times \cdots \times \vartheta_r$, it checks if
$\vartheta$ includes the grammar with the production for the single node $k$, the key, and distinct symbol $\textsc{Enc}^\ell_r$, with arity $r$, generating the trees rooted in $enc^\ell_r$ with subtrees $\hat{v}_1,\cdots,\hat{v}_r$.
The premises of the rule (P-dec)  
(i) check if each $\vartheta_i$ approximates the values of each term $E_i$ involved in the matching;
(ii) inspect the approximation $\vartheta$ of $E$ to be decrypted and from each $\hat{v}$ in this abstract value it extracts the set of ordered lists L = $[\hat{v}_1, \dots, \hat{v}_r]$ of the grammars associated with the sub-terms of $E$ through the auxiliary function $\mathtt{D}$;
(iii) check that the presence of such a list L implies that the values $\hat{v}_i$ that at run time are bound to $x_i$ ($i \in [j+1,r]$) are correctly predicted (i.e.\ belong to $\widehat{\Sigma}{_\ell}({x_i})$), and in addition that the estimate in hand is valid for the continuation $P$ (note that this last check is only done if there exist a non-empty list L). 
If all the above holds, the decryption process is correctly analysed.

In order to establish the subject reduction for proving the correctness of this extension to our CFA, we extend the instrumented denotational semantics for terms in the following expected way:
\[
\dsem{\{E_1,\cdots,E_r\}_{k}}^i_{\Sigma^{i}_{\ell} } = 
(\{\dsem{E_1}^i_{\Sigma^{i}_{\ell} \downarrow_1},\cdots, \dsem{E_r}^i_{\Sigma^{i}_{\ell} \downarrow_1}\}_{k}, \;
enc^\ell_r(\dsem{E_1}^i_{\Sigma^{i}_{\ell} \downarrow_2},\cdots,\dsem{E_r}^i_{\Sigma^{i}_{\ell} \downarrow_2} ,k))
\]

Now we simply assert that the theorems and the corollary of the previous sections still hold, without formally writing them, as nothing changes in their statements.
Their proofs have already been included in those of the original versions.

%%%%%%%%%%% NUOVO %%%%%%%%%
\subsection{Controlling data propagation}

We now introduce our general schema to express security policies and a way of checking them at static time.
As said above, we provide the designer with ways to classify data and nodes and to regulate how data are propagated within the system.
In the next sub-sections we will illustrate this schema by instantiating it on three specific cases.

In order to classify data, the designer associates with them a security tag, taken from a finite set $\mathfrak{D}$.
Then, it defines a pair of functions that accordingly assign tags to abstract values.
As expected, all the trees generated by a grammar must have the same tag that is also assigned to the grammar.

\begin{definition}[Tagging data]\label{def:tagging}
Given a finite set $\mathfrak{D}$ of \emph{tags}, the functions
\[
\mathfrak{T_D} : \widehat{\mathcal{T}} \rightarrow \mathfrak{D} 
\qquad\qquad\qquad
\mathfrak{T_S} : \widehat{\mathcal{V}} \rightarrow \mathfrak{D} 
\]
are \emph{tagging functions} if and only if, given $\widehat{G} \in \widehat{\mathcal{V}}$, for all $\hat{t} \in Lang(\widehat{G})$ the following condition holds
\[
\mathfrak{T_D}(\hat{t}) = \mathfrak{T_S}(\widehat{G})
\]

\end{definition}

The other element of the schema is a policy $\mathfrak{P}$.
In our case, it is simply defined as a predicate over tags and labels.
Intuitively, a designer dictates through a policy if data with a specific tag can be exchanged between two nodes and how, e.g.~encrypted or as clear text.
Our schema is formalised below.

\begin{definition}[Data propagation policies]\label{def:data-policy}
Let $N$ be a system of nodes with labels in $\mathcal{L}$; let $\mathfrak{D}$ be a finite set of tags; let $\mathfrak{T_D} : \widehat{\mathcal{T}} \rightarrow \mathfrak{D}$; and let
$\mathfrak{P} \subseteq \mathfrak{D} \times \mathcal{L} \times \mathcal{L}$ be a data propagation policy.
\\
Then \emph{$N$ enjoys $\mathfrak{P}$} %with respect to $\mathfrak{D}$} 
if and only if $N \rightarrow^* N'$ and for all 
$\ell_1, \ell_2 \in \mathcal{L}$
there is no transition such that 
$N' \xrightarrow{\mess{v_1,\dots,v_r}}_{\ell_1, \ell_2} N''$
and $\mathfrak{P}(\mathfrak{T_D}(v_{i \downarrow 2}), \ell_1, \ell_2)$ does not hold, for some $i \in [1,r]$.

\end{definition}

The following theorem guarantees that the component $\kappa$ of the analysis can be inspected to statically check if a given system of nodes enjoys the wanted security policy.

\begin{restatable}[Well propagation]{thm}{propaga}\label{Th:propagated}
Let $N$ be a system of nodes with labels in $\mathcal{L}$; let $\mathfrak{D}$ be a finite set of tags; 
let $\mathfrak{T_D}, \mathfrak{T_S}$ be a pair of tagging functions; and let
$\mathfrak{P} \subseteq \mathfrak{D} \times \mathcal{L} \times \mathcal{L}$.
\\
Then $N$ enjoys $\mathfrak{P}$ if %with respect to $\mathfrak{D}$ if
\begin{enumerate}
\item
$\nNAFORM{N}$ and
\item
$\forall \ell_1, \ell_2 \in {\mathcal L}$ if $(\ell_1,  \mess{\hat{v}_1,\cdots,\hat{v}_r}) \in  \kappa(\ell_2)$ 
then $\forall \, i \in [1,r]$ it holds $\mathfrak{P}(\mathfrak{T_S}(\hat{v}_i), \ell_1, \ell_2)$.
\end{enumerate}
\end{restatable}

Note that a more general version of the above can be obtained by turning the set $\mathfrak{D}$ into a finite complete lattice $(\mathfrak{D}, \sqsubseteq)$; 
requiring that the tagging functions satisfy $\mathfrak{T_D}(\hat{t}) \sqsubseteq \mathfrak{T_S}(\widehat{G})$ rather than equality; 
and imposing that $\mathfrak{P}$ is monotone in its first argument.

The next two sub-sections instantiate the above schema on two classical security properties, confidentiality and no-read-up/no-write down, while the third sub-section instantiates the above schema the security property discussed in the example of Section~\ref{sec:example}.
As we will see, confidentiality only requires classifying data and imposes no constraint on nodes.
Instead, the second property neglects the tags of data and only considers clearance levels of nodes.
The last one has constraints on both data and the way these are propagated within the system.

\subsection*{Preventing Leakage}

As the first example of how the notion of well propagation is instantiated, we consider the common approach of identifying the sensitive content of information and of detecting possible disclosures.
Hence, we partition values
into security classes and prevent classified information from flowing in clear or to the wrong places.
In the following, we assume that the designer defines the set ${\mathcal S}_{\ell}$ containing the sensors of the node $N$ with label ${\ell}$ that have to be protected, i.e.\ the values of which are to be kept {\em secret}.
This implicitly introduces also the set ${\mathcal P}_{\ell}$ of
{\em public} sensors as the complement of ${\mathcal S}_{\ell}$.
In this case, the set of tags $\mathfrak{D}$ is $\{secret, public \}$. 

The abstract values are partitioned through the two tagging functions $\mathscr{D}$ and $\mathscr{S}$ defined below. 
The intuition behind them is that a single ``drop'' of $\s$ turns to \emph{secret} the (abstract) term.
Of course there is the exception for encrypted data:
what is encrypted is public, even if it contains secret components.
Besides classifying sensors, one can obviously consider the case of aggregation functions whose result has to be kept secret.
%%\textcolor{cyan}{A different treatment might also consider the case 
%%when applying an aggregation function to public values produces confidential data, but a suitable 
%%anonymisation function applied to it makes it public again (see example...).
%%Given the set of such functions, a slight change in the following definition will suffice.
%%}

\begin{definition}\label{kinds}
Given the sets ${\mathcal S}_{\ell}$ and ${\mathcal P}_{\ell}$ of {\em secret} and of \emph{public}
sensors of the node $N$ with label $\ell$, we define the pair of functions
\\
%$\bullet\ $ the {\em dynamic} operator $\mathscr{D}: \widehat{\mathcal{T}} \rightarrow \{\s, \p\}$ as follows:
$\bullet\ \mathscr{D}: \widehat{\mathcal{T}} \rightarrow \{\s, \p\}$ as follows:

\[
\begin{array}{ll}
\mathscr{D}(i^\ell) = 
\begin{cases}
\s & \mbox{ if } i  \in {\mathcal S}_{\ell}
\\
\p & \mbox{ if } i \in {\mathcal P}_{\ell}
\end{cases}
&
\hspace{-6mm} 
\mathscr{D}(v^\ell) = public
%\begin{cases}
%\s  & \mbox{ if } v  \in {\mathcal S}_{\ell}
%\\
%\p  & \mbox{ if } v  \in {\mathcal P}_{\ell}
%\end{cases}
\\ \  \\

\mathscr{D}(encr^\ell_r(\hat{t}_1,\cdots, \hat{t}_r, k)) = \p
\\ \ \\
\mathscr{D}(f^\ell(\hat{t}_1,\ldots,\hat{t}_r)) = 
    \begin{cases}
    \s & \mbox{ if } \exists\, \hat{t}_i \text{ s.t. } \mathscr{D}(\hat{t}_i) = \s
    \\
    \p  & \text{ otherwise}
    \end{cases}
&
\end{array}
\]
%$\bullet\ $ the {\em static} operator $\mathscr{S}: \mathcal{\widehat{V}} \rightarrow \{\s, \p\}$ as follows:
$\bullet\ \mathscr{S}: \mathcal{\widehat{V}} \rightarrow \{\s, \p\}$ as follows:

\[
\mathscr{S}(\hat{v}) =
\begin{cases}
\s & \text{ if } \exists (Z, R) \in \hat{v} \text{ s.t. } \mathscr{S^*}((Z, R)) = \s
\\
\p & \text{ otherwise}
\end{cases}
\]
where $\mathscr{S^*}: \mathbb{N} \times \mathbb{R}  \rightarrow \{\s, \p\}$ is:
\[
\begin{array}{l}
\mathscr{S^*}((I^{\ell}, R)) = 
\begin{cases}
\s  & \mbox{ if }i^{\ell}  \in {\mathcal S}_\ell
\\
\p & \mbox{ if }i^{\ell}  \in {\mathcal P}_\ell
\end{cases}
\qquad\qquad\qquad\qquad\qquad
\mathscr{S^*}((V^{\ell}, R)) = public
\\ \ \\
\mathscr{S^*}((\textsc{Enc}^\ell_{k}, R)) = \p
\\ \ \\
\mathscr{S^*}((F^\ell, \{F^\ell \rightarrow f^\ell(Z_1, \cdots, Z_r) \} \, \cup\, R)) = 
\begin{cases}
\s & \text{ if } \exists (Z_i, R_i) \text{ s.t. } \mathscr{S^*}((Z_i, R_i)) = \s
\\
\p & \text{ otherwise}
\end{cases}
\end{array}
\]
\end{definition}
%\begin

\begin{restatable}{fatto}{sed}\label{lemma:S-e-D}
%Given a grammar $\widehat{G}$ and an abstract tree $\hat{t} \in Lang(\widehat{G})$, then
%$\mathscr{D}(\hat{t}) = \mathscr{S}^*(\widehat{G})$.
$\mathscr{D}$ and $\mathscr{S}$ are a pair of tagging functions.
\end{restatable}

Since our analysis computes information on the values exchanged during the communication, we can statically check whether a value, devised to be secret to a node $N$, is never sent
to another node, similarly to~\cite{BBDNN_InfComp}.
We first give a dynamic characterisation of when a node $N$ never discloses its secret values, i.e.\ when neither it nor any of its derivatives can send a message that includes a secret value (recall that 
$N \xrightarrow{\mess{v_1,\dots,v_r}}_{\ell_{1}, \ell_{2}} N'$ stands for sending the message from $\ell_1$ to $\ell_2$).
\begin{definition}
Let $N$ be a system of nodes with labels in $\mathcal{L}$, and let  
$\mathcal{S} = \{\mathcal{S}_{\ell} \mid \ell \in \mathcal{L} \}$ be the set of its secret sensors.
\\
Then $N$ has {\em no leaks} with respect to $\mathcal{S}$ if and only if $N \rightarrow^* N'$ and
for all $\ell_1, \ell_2 \in \mathcal{L}$
there is no transition
$N' \xrightarrow{\mess{v_1,\dots,v_r}}_{\ell_1, \ell_2} N''$
such that $\mathscr{D}(v_{i {\downarrow_2}}) = \s$ for some $i \in [1,r]$.
\end{definition}
The definition above instantiates the predicate $\mathfrak{P}$ by imposing that the first argument has to be $public$ and it simply ignores the labels of the nodes.
The component $\kappa$ allows us to statically  detect when a system of nodes $N$ has no leaks with respect to a given set of secret sensors, as an instance of Theorem~\ref{Th:propagated}. 

\begin{restatable}[Confidentiality]{thm}{confine}\label{Th:confinement}
Let $N$ be a system of nodes with labels in $\mathcal{L}$, and let  
$\mathcal{S} = \{\mathcal{S}_{\ell} \mid \ell \in \mathcal{L} \}$ be the set of its secret sensors.
Then $N$ has no leaks with respect to ${\mathcal S}$ if
\begin{enumerate}
\item
$\nNAFORM{N}$ and
\item
$\forall \ell_1, \ell_2 \in {\mathcal L}$ if $(\ell_1,  \mess{\hat{v}_1,\cdots,\hat{v}_r}) \in  \kappa(\ell_2)$ 
then $\forall \, i \in [1,r]$ $\mathscr{S}(\hat{v}_i) = \p$.
\end{enumerate}
\end{restatable}

Back to our running example of Section~\ref{sec:example}, we said that the pictures of cars are sensitive data, and that consequently one would like to keep them secret.
To check this, we classify data coming from the sensor $S_{cp}$ such that $\mathscr{D}(1^{cp}) = \mathscr{S}(\iota) = secret$ (recall the definition of $\iota$ in~(\ref{eq:iota})).
Accordingly to Definition~\ref{kinds}, we also get $\mathscr{D}(noiseRed^{cp}(1^{cp})) = \mathscr{S}(\nu) = secret$ (recall $\nu$ from~(\ref{eq:xi})).
Hence, by simply inspecting $\kappa$ 
we discover that the sensitive data of cars could be sent in clear to the street supervisor $\ell_a$, so violating secrecy: indeed $\kappa(\ell_a) \supseteq \{ (\ell_{cp},\mess{\nu }\}$.
We illustrated an amended system in Section~\ref{sec:example} that does not suffer from this problem.

%\textcolor{red}{
%Back to our running example, since the pictures of cars are sensitive data, one would like to check whether they are kept secret.
%By classifying $1^{\ell_{cp}}$ as one of the secret elements for the node $\ell_{cp}$, we accordingly get that $noiseRed^{^{\ell_{cp}}}(1^{\ell_{cp}})$ is $\s$.
%By
%inspecting $\kappa$ 
%we discover e.g.\  that the sensitive data of cars is sent in clear to the street supervisor $\ell_a$, so possibly violating privacy; indeed $\kappa(\ell_a) \supseteq \{ (\ell_{cp},\mess{noiseRed^{^{\ell_{cp}}}(1^{\ell_{cp}})})\}$.
%%
%To prevent this disclosure the pictures are better sent encrypted.
%}

\subsection*{Communication policies.}
We present a second example of instantiation of the general schema above.
%It considers the case when security is enforced by defining policies on communications that rule information flows among nodes, allowing some flows and forbidding others.
It considers the case when security is enforced by defining policies
that rule information flows among nodes, by allowing some flows and forbidding others.

Below we consider the well-known {\em no read-up/no write-down}  policy~\cite{BL73,Den82}.
It is based on a hierarchy of clearance levels for nodes, and
it requires that a node classified at a
high level cannot write any value to a node at a lower level,
while the converse is allowed; symmetrically a node at low level
cannot read data from one of a higher level.

For us, it suffices to classify the node labels with an assignment function $level : {\mathcal L} \rightarrow {\bf L}$, from the set of node labels to a given set of levels $\bf L$.
We then introduce a condition for characterising the allowed and the forbidden flows:
it suffices requiring that whenever a message is sent by to a node $\ell_1$ to a node $\ell_2$, their levels satisfy the condition
$
level(\ell_1) \leq level(\ell_2)
$,
as defined below.
We instantiate the schema as follows: we take a singleton set for $\mathfrak{D}$; a pair of constant functions as tagging functions; and a predicate $\mathfrak{P}$ that ignores the first argument and checks the required disequalities between levels.

\begin{definition}
Given an assignment function \emph{level}, and a system of nodes $N$ with labels in $\mathcal{L}$, 
then $N$ {\em respects the levels} if and only if $N \rightarrow^* N'$ and for all $\ell_1, \ell_2 \in \mathcal{L}$
there is no transition such that 
$N' \xrightarrow{\mess{v_1,\dots,v_r}}_{\ell_1, \ell_2} N''$
and $level(\ell_1) \leq level(\ell_2)$.

\end{definition}
Again the component $\kappa$ of the analysis can be used to statically predict when the components of a system of nodes will communicate respecting the assigned levels of clearance, as an instance of Theorem~\ref{Th:propagated}.

\begin{restatable}[No-read-up/no-write-down]{thm}{nrunwd}\label{Th:nru-nwd}
Given an assignment function \emph{level}, and a system of nodes $N$ with labels in $\mathcal{L}$.
Then $N$ {\em respects the levels} if
\begin{enumerate}
\item
$\nNAFORM{N}$ and
\item
$\forall \ell_1, \ell_2 \in {\mathcal L}$ if $(\ell_1,  \mess{\hat{v}_1,\cdots,\hat{v}_r}) \in  \kappa(\ell_2)$ 
then
$level(\ell_1) \leq level(\ell_2)$.
\end{enumerate}
\end{restatable}

%Of course, one can mix the above checks to verify a composition of the considered properties, e.g.~by checking whether a particular secret value does not flow to a specific node, even if it has a higher level.

More in general, we can constrain communication flows according to a specific policy, by stating whom a node is allowed to send (and/or from which it can receive) a message.
It suffices to define a relation $\mathbf{R} \subseteq \mathcal{L} \times 2^\mathcal{L}$, in place of the relation used above that is based on the partial ordering between the levels of nodes.

%%%\textcolor{red}{
%%%In our running example, we can check (the trivial fact) that the communication from node $N_a$ to $N_s$ is allowed by the policy, while those in the other direction are not. 
%%%}

%Also, note that a pretty similar approach can be used to deal with trust among nodes: instead of assigning nodes a security level, just give them a discrete measure of their trust.
%Similarly, we could check whether a node is allowed to use a particular aggregation function.

%\textcolor{red}{ci frega del trust?}

\subsection*{Selective data propagation}
We consider now a form of access control policy that confines exchange of a specific piece of information within a group of nodes in a system.
An example of such a policy has been intuitively discussed at the end of Section~\ref{sec:example}.
It dictates that the picture of an incoming car must only be sent to $N_a$, the supervisor of the camera, and to $N_s$, the supervisor of the lamp posts, but not to the lamp posts $N_p$, if the picture has not been anonymised.

In order to express the above policy in our general schema, we follow the pattern of the above examples.
Assume to have the set ${\mathcal R}_{\ell}$ containing the sensors of the node $N$ with label ${\ell}$, the values of which are to be kept {\em confined}, together with its complement, the \emph{open} values.
In this case, the set $\mathfrak{D}$ is $\{\mathit{confined}, open \}$. 
%It is straightforward to have a hierarchy of classification levels associated with values, instead of just the two above.

The abstract values are partitioned through the two tagging functions $\mathscr{C}$ and $\mathscr{O}$ defined below. 
The intuition behind them is that a single ``drop'' of  \emph{confined} in a term makes such the whole, with the exception when the anonymisation function is applied:
what is anonymised is \emph{open}, even if it contains confined elements.

\begin{definition}\label{kinds2}
Let ${\mathcal R}_{\ell}$ be the set of {\em confined} sensors of the node $N$ with label $\ell$, and   
$\tilde{{\mathcal F}} \subseteq \mathcal{F}$ be the set of anonymisation functions, possibly including encryptions.
Then we define the following pair of functions
\medskip\\
$\bullet\ \mathscr{C}: \widehat{\mathcal{T}} \rightarrow \{\mathit{confined}, open\}$ as follows:
\[
\begin{array}{ll}
\mathscr{C}(i^\ell) = 
\begin{cases}
\mathit{confined} & \mbox{ if } i  \in {\mathcal T}_{\ell}
\\
open & \mbox{ otherwise }
\end{cases}
&
\hspace{-36mm} 
\mathscr{C}(v^\ell) = open
%\begin{cases}
%\s  & \mbox{ if } v  \in {\mathcal S}_{\ell}
%\\
%\p  & \mbox{ if } v  \in {\mathcal P}_{\ell}
%\end{cases}
\\ \  \\

%\mathscr{C}(encr^\ell_r(\hat{t}_1,\cdots, \hat{t}_r, k)) = \begin{cases}
%    \mathit{confined} & \mbox{ if } \exists\, \hat{t}_i \text{ s.t. } \mathscr{C}(\hat{t}_i) = \mathit{confined} \ \land \ f^\ell  \not\in \tilde{{\mathcal F}}
%\\
%open & \mbox{ otherwise }
%\end{cases}
%\\ \ \\
\mathscr{C}(f^\ell(\hat{t}_1,\ldots,\hat{t}_r)) = 
    \begin{cases}
    \mathit{confined} & \mbox{ if } \exists\, \hat{t}_i \text{ s.t. } \mathscr{C}(\hat{t}_i) = \mathit{confined} \ \land \ f  \not\in \tilde{{\mathcal F}}
    \\
    open  & \text{ otherwise}
    \end{cases}
&
\end{array}
\]
$\bullet\ \mathscr{O}: \mathcal{\widehat{V}} \rightarrow \{\mathit{confined}, open\}$ as follows:
\[
\mathscr{O}(\hat{v}) =
\begin{cases}
\mathit{confined} & \text{ if } \exists (Z, R) \in \hat{v} \text{ s.t. } \mathscr{O}^*((Z, R)) = \mathit{confined}
\\
open & \text{ otherwise}
\end{cases}
\]
where $\mathscr{O}^*: \mathbb{N} \times \mathbb{R}  \rightarrow \{\mathit{confined}, open\}$ is:
\[
\begin{array}{l}
\mathscr{O^*}((I^{\ell}, R)) = 
\begin{cases}
\mathit{confined}  & \mbox{ if }i^{\ell}  \in {\mathcal T}_\ell
\\
open & \mbox{ otherwise }
\end{cases}
\hspace{2cm}
\mathscr{O^*}((V^{\ell}, R)) = open
\\ \ \\
%\mathscr{O^*}((ENC^\ell_{k}, R)) = \p
%\\ \ \\
\mathscr{O^*}((F^\ell, \{F^\ell \rightarrow f^\ell(Z_1, \cdots, Z_r) \} \, \cup\ R)) = 
\\ \hspace{2cm}
\begin{cases}
\mathit{confined} & \text{ if } \exists (Z_i, R_i) \text{ s.t. } \mathscr{O^*}((Z_i, R_i)) = \mathit{confined}
          \ \land \ f  \not\in \tilde{{\mathcal F}}
\\
open & \text{ otherwise}
\end{cases}
\end{array}
\]
\end{definition}

\begin{restatable}{fatto}{ceo}\label{lemma:C-e-O}
$\mathscr{C}$ and $\mathscr{O}$ are a pair of tagging functions.
\end{restatable}

The following definition formalises when a system of nodes propagates confined data only within the selected sub-system. 
For that, we let the set $\tilde{\mathcal{L}} \subseteq \mathcal{L}$ to include such nodes.
\begin{definition}
Let $N$ be a system of nodes with labels in $\mathcal{L}$, and let  
$\mathcal{R} = \{\mathcal{R}_{\ell} \mid \ell \in \mathcal{L} \}$ be the set of its confined sensors.
\\
Then $N$ {\em selectively propagates confined data} with respect to $\mathcal{R}$ and $\tilde{\mathcal{L}}$ if and only if $N \rightarrow^* N'$ and
for all $\ell_1, \ell_2 \in \mathcal{L}$
there is no transition
$N' \xrightarrow{\mess{v_1,\dots,v_r}}_{\ell_1, \ell_2} N''$
such that $\mathscr{O}(v_{i {\downarrow_2}}) = \mathit{confined}$ and 
($\ell_1 \not\in \tilde{\mathcal{L}}$ or $\ell_2 \not\in \tilde{\mathcal{L}}$), for some $i \in [1,r]$.
\end{definition}

Once again, the component $\kappa$ allows us to statically detect when a system of nodes $N$ selectively propagates confined data within a sub-system, as an instance of Theorem~\ref{Th:propagated}. 

\begin{restatable}[Selective data propagation]{thm}{selectiveprop}\label{Th:selective-prop}
Let $N$ be a system of nodes with labels in $\mathcal{L}$; let  
$\mathcal{R} = \{\mathcal{R}_{\ell} \mid \ell \in \mathcal{L} \}$ be the set of its confined sensors;
and let $\tilde{\mathcal{L}} \subseteq \mathcal{L}$ be a sub-system.
Then $N$ selectively propagates confined data with respect to $\mathcal{R}$ and $\tilde{\mathcal{L}}$ if\begin{enumerate}
\item
$\nNAFORM{N}$ and
\item
$\forall \ell_1, \ell_2 \in {\mathcal L}$ if $(\ell_1,  \mess{\hat{v}_1,\cdots,\hat{v}_r}) \in  \kappa(\ell_2)$ 
then whenever $\mathscr{O}(\hat{v}_i) = \mathit{confined}$ for some  $i \in [1,r]$, 
it is $\ell_1, \ell_2 \in \tilde{\mathcal{L}}$.
\end{enumerate}
\end{restatable}

Back to our example, let the set of confined sensors $\mathcal{R}$ be the singleton $\{S_{cp}\}$, the set of allowed nodes $\tilde{\mathcal{L}}$ be 
$\{\ell_{cp},\,\ell_{a},\,\ell_{pd}\}$ and the set of anonymisation function $\tilde{\mathcal{F}}$ be $\{an\}$.
From our analysis we know that $\kappa(\ell_s) \in (\ell_a, \mess{X,Y})$, where $(\textsc{Car}, \{\textsc{Car} \rightarrow car\}) \in X$ 
and $\nu \in Y$ where $\nu$ is defined in~(\ref{eq:xi}).
We discover a policy violation since $\mathscr{O}(Y) = \mathit{confined}$ but $\ell_s \notin \tilde{\mathcal{L}}$.
Now consider our example with the node $N'_a$, introduced at the end of Section~\ref{sec:example}.
In this case we have no policy violation because,
from our analysis, we obtain that $\kappa(\ell_s) \in (\ell'_a, \mess{X,Z})$ where $X$ is as above, and $Y$ includes
\[
(\textsc{An}, \{ \textsc{An} \rightarrow an^{\ell'_a}(\textsc{NoiseRed}^{\ell_{cp}}), \textsc{NoiseRed}^{\ell_{cp}} \rightarrow noiseRed^{\ell_{cp}}(M^{\ell_{cp}}), M^{\ell_{cp}} \rightarrow 1^{\ell_{cp}}\} )
\]
and $\mathscr{O}(Y) = open$.

Above, we only have two tags, but one can easily adapt the definition taking into account a set of tags.
Accordingly, the system of nodes can be partitioned in many sub-systems.
It is easy to constrain the propagation of a piece of information with a certain tag to only occur within given sub-nets.
A further generalisation refines the data propagation policy by defining for each node $N_\ell$ and each tag $d$ the set of nodes to which $N_\ell$ is allowed to send data tagged with $d$.

%% file: sections/concl.tex
% !TEX root = ../main_coord2016.tex

We proposed the process calculus \IoTLySa\ as a formal design framework for IoT.
It aims at supporting the designer of IoT applications and at providing techniques for checking their properties.
\IoTLySa\ has constructs to describe sensors and actuators, and suitable primitives for managing the coordination and the communication capabilities of  interconnected smart objects.
We equipped our calculus with a Control Flow Analysis that statically computes a safe over-approximation of system behaviour.
In particular it predicts how nodes interact, how data spread from sensors to the network, and how data are processed.

These approximations offer the basis for checking various properties of IoT systems, among which the classical security properties of secrecy and no-read-up/no-write-down. 
In addition, we considered a combination of policies ruling access control and data propagation.
In order to better assess the feasibility of our proposal, we are implementing a tool to assist the designer in defining the properties of interest, e.g.\ following the schema of Section~\ref{sec:security}, and in checking them on the system at hand.
Prior to that is the computation of the (minimal) valid estimate.
This requires deriving a set of constraints from the analysis specification and using the succinct solver of~\cite{NielsonAFPL}.

In order to  experiment on the usability of \IoTLySa\ as a specification language we are currently considering a few case studies.
In a scenario of precision agriculture, we are modelling an efficient and
sustainable irrigation system to manage water in a vineyard~\cite{Vineyard-2017}.
Another example considers an IoT system that controls the temperature in a smart storehouse with particular attention to countermeasures in case an attacker tampers with sensors~\cite{BG_ICTCS16}.
We recently started a collaboration with Zerynth, a spin-off of the Pisa University\footnote{\url{https://www.zerynth.com/}} for designing a smart fridge system.

Further future work concerns the extension of \IoTLySa\ itself.
A first issue is considering mobility of smart objects.
A simple solution would be making dynamic the topology of the network through suitable primitives that update the compatibility function $Comp$ at run time.
However, the accuracy of our present analysis would be certainly affected by the arising non-determinism.
A richer solution would be representing also spatial information in the code and providing the designer with primitives that specify the movement of objects from one location to another. 

Addressing mobility makes it evident that smart objects should adapt their behaviour according to the different features of the location hosting them.
Indeed, in the IERC words: ``there is still a lack of research on how to adapt and tailor existing research on autonomic computing to the specific characteristics of IoT''~\cite{IERC}. 
To contribute to these issues, we plan to extend our calculus with linguistic mechanisms and  a verification machinery to deal with adaptivity in the style of~\cite{ieee16,JCS-2016}.
%Such an extension will enable a designer to deal with situations like the following.
%Back to our running example, suppose to add further sound sensors to the lamp posts in order to detect loud noises like screams.
%The system should adapt its behaviour according to the time of the day: it turns on the lights at night, and leave them switched off otherwise.

As it is, \IoTLySa\ assumes that the components of a node never fail.
To relax this assumption, a quantitative semantics will help in taking into account the probability of failure of each node component. 
Of course, the analysis and the verification will be deeply affected by this new semantics.
Actually, a quantitative semantics would support tuning the design of a system.
For example, a designer can better shape the dimension of a node with respect to energy consumption, or balance the computational load of nodes by allocating specific data aggregations within the system. 
A quantitative semantics would also help in studying the trade-off between the costs of protecting some assets in IoT system and the damage an attacker can cause.
%probability of certain attacks to the security and the privacy of an IoT system, besides evaluating the costs of possible counter-measures (see the discussion of~\cite{HotSpot} below). 
%\\
%************
%C'è un problema: la semantica di HotSpot serve proprio a questo.
%Come si riconcilia?
%Forse si può rifrasare parlando di allocazione di security mechanisms citando bertino2.pdf 
%***********

For the time being, \IoTLySa\ models only the logical part of a system, leaving the physical world as a black box. 
A further extension consists in opening this box and in modelling the world through standard techniques as e.g.\ hybrid automata. 
This extension not only improves the precision of \IoTLySa\ models but also could allow designers to understand how sensors are used and how an action triggered on an actuator affects the whole system.
As an example, one can predict if an actuator is maliciously triggered by an attacker, as happened in the recent attack performed through a vehicular infotainment 
network.\footnote{See \url{http://www.wired.com/2015/07/hackers-remotely-kill-jeep-highway/}}

%%%%%%%%%%%%%%%%%%%%%%%%%%%%%%%%%%%%%%%

\subsection{Related work}
%Formal methods techniques need to be re-thought and reshaped to account for the new challenges IoT introduces.
As  all new paradigms, IoT poses new challenging scenarios for formal methodologies.
%Still its formal modelling is poorly addressed in the literature, and, to 
To the best of our knowledge, few papers address the specification and verification of IoT systems from a process calculi perspective.

A first proposal is the IoT-calculus~\cite{lanese13}.
It explicitly includes sensors and actuators, and smart objects are represented as point-to-point communicating nodes of  networks, represented as a graph.
Differently from ours, their interconnection topology can vary at run time, when the system interacts with the environment.
This is rendered by a first kind of transitions, the others taking care of the node activities.
The authors propose two notions of bisimilarity, each based on the different kinds of transitions, that capture system behaviour from the point of view of end-users and of the selected devices.
The clean algebraic presentation enables compositional reasoning on systems.

%%% DESCRIZIONE DI MERRO DEL PAPERO LANESE
%To our knowledge, paper [19] is the only one proposing a process calculus for IoT systems to
%capture the interaction between sensors, actuators and computing processes. Smart objects are
%represented as point-to-point communicating nodes of heterogeneous networks. The network
%topology is represented as a graph whose links can be nondeterministically established or
%destroyed. The paper contains a labelled transition system with two different  kinds of
%transitions. The first one takes into account interactions with the physical environment,
%similarly to our physical transitions, but includes also topology changes. The second kind of
%transition models nodes activities, mainly communications, similarly to our logical transitions.
%Then the paper proposes two notions of bisimilarity: one using only the first kind of transitions
%and equating systems from the point of view of the end user, and a second one using all
%transitions and equating systems from the point of view of the other devices. Only the second
%form of bisimilarity is preserved by parallel composition.
%Compared to [19], CIoT introduces a notion of discrete
%

The timed process calculus {\tt CIoT}~\cite{Merro_Coord16} specifies physical and logical components, 
addresses both timing and topology constraints, and allows for node mobility.
Furthermore, communications are either short-range or internet-based.
The semantics is given on terms of a labelled transition system with transitions
that represent the interactions with the physical environment and transitions that represent
the node actions.
The focus of this paper is mainly on an extensional semantics that 
emphasises the interaction of IoT systems with the environment, and that
provides a \emph{fully abstract} characterisation
of the proposed contextual equivalence.
The hybrid process calculus CCPS~\cite{LanotteM17} describes both the cyber and the physical aspects of systems (state variables, physical devices, evolution laws, etc.). 
It has a clearly defined behavioural semantics, based on a transition system and a notion of bisimulation.
As in our case, a cyber component governs the interactions with sensors and actuators,
and the communications with other cyber components.
The physical world is modelled by a tuple of suitable functions over the real numbers that describe how it changes.
Many design choices of the above-discussed proposals are similar to ours. 
The main difference is that our coordination model is based on a shared store \`a la Linda instead of a message-based communication \`a la $\pi$-calculus.
Furthermore, differently from~\cite{lanese13,Merro_Coord16}, 
we are here mainly interested in developing a design framework that includes a static semantics to support various verification techniques and tools for checking properties of IoT applications.
In addition, as said above, here we are only interested in modelling the logical components, in determining their boundaries with the physical world, and in abstractly representing the interactions between them. Modelling the physical environment is left as future work, e.g.\ along the lines suggested in~\cite{LanotteM17}.
Finally, we opted for an asynchronous multi-party communication modality among nodes, while in IoT-calculus and in {\tt CIoT} internode communications are synchronous and point-to-point.

The calculi above and ours are built upon those previously introduced for wireless, sensor and ad hoc networks
(\cite{lanese10,singh10,NNN10} to cite only a few).
In particular, the calculus in~\cite{NNN10} is designed to model so-called broadcast networks, with a dynamically changing topology.
It presents some features very similar to ours:
an asynchronous local broadcast modality, while intra-node communication relies on a local tuple space.
Also, the analysis of the behaviour of broadcast networks is done by resorting to a multi-step static machinery.
%
%the work that is nearest to ours is the one in~\cite{NNN10}, which has a similar local broadcast communication pattern, centred 
%around the tuple space paradigm...CONTINUARE
%CITARE ANCHE BioPEPA e  collective behaviour o Collective Adaptive Systems?.
%

Also {\em aggregate programming} has been proposed 
for the design and implementation of complex IoT software systems~\cite{BealPV15}, through the {\em Field calculus}~\cite{DamianiVPB15}, a functional language with constructs to model computation and interaction among  large numbers of devices. 
Aggregate programming overcomes the single-device viewpoint, by adopting a cooperating collection of devices as basic computing unit.
Also the rule-based approach of~\cite{ECA} describes a network of sensors and actuators as a distributed collaborative system.
It chiefly focusses on globally coordinating the activities of the devices and on controlling the events raised by sensors and the effects of the actions of actuators they trigger.
Our approach is complementary to theirs, because we are mainly interested in how the data gathered from sensors are communicated and subsequently aggregated  by the nodes in the network.
The goals of ThingML~\cite{thingml} are similar to ours in that its authors propose a modelling language for IoT.
However, our work is oriented towards formal verification, while theirs is devoted to tools for supporting the development of applications, within a Software Engineering approach.

Security in IoT and cyber-physical systems has been addressed from a process algebraic perspective by a certain number of papers, e.g.~\cite{NielsonZigB,VigoNN13,Aziz16}.
An IoT protocol (the MQ Telemetry Transport 3.1) is formally model in terms of a timed message-passing process algebra~\cite{Aziz16}. 
The protocol is statically analysed to understand the robustness of its behaviour in scenarios with different quality of service.
The Applied Quality Calculus is used in~\cite{VigoNN13} to study the trade-off between the security requirements and the broadcast communications in wireless-based Cyber-Physical Systems, where nodes have limited capabilities and sensors are poorly reliable.
In this proposal, standard cryptographic mechanisms are modelled via term rewriting. 
Furthermore it is possible to reason about denial of service, because the calculus has explicit notions of failing and unwanted communications.

Very similar to ours is the work in~\cite{NielsonZigB}, where the key establishment protocol of ZigBee~\cite{ZigBee} is specified in {\LySa}~\cite{BBDNN_JCS,BNN03} and statically analysed with a CFA.
They discovered a security flaw arising because freshness is not guaranteed, and proposed a fix.

A  variant of the CFA presented here is proposed in~\cite{BG_ITASEC17} for tracking sensitive and untrustworthy data.
Taint analysis is used for marking data and for monitoring their propagation at run time across an IoT system so as to check whether those computations considered security critical are not affected by tainted data.

Finally, security is studied in~\cite{BG_ICTCS16} from an ``economical'' perspective, by providing an enhanced version of \IoTLySa\ to infer quantitative measures of cryptographic mechanisms.
This paper provides the means for estimating the costs of possible counter-measures to several security risks that may arise in IoT systems.
A game theoretical approach to the same problem is in~\cite{Bertino2}.
The authors compute a resource allocation plan as a Pareto-optimal solution, and estimate its efficiency with respect to the energy consumption of the security infrastructure and the costs of its components.

%% file: sections/notation.tex
% !TEX root = ../main_coord2016.tex

\footnotesize{
\begin{table}[htp]
\begin{center}
\begin{tabular}{|l|l||l|l|}
\hline
$N \in \mathcal N$ & system of nodes N & $v \in {\mathcal V}$ & values\\
$\ell \in {\mathcal L}$ & the set of labels  & $x \in {\mathcal X}_\ell$ & variables of node $\ell$ \\
$\ell: [B]$ & node at $\ell$ with component B & $f \in {\mathcal F}$ & functions on terms E\\
S, A & sensors and actuators & $\dsem{E}_\Sigma$ & denotational semantics of terms\\
$i \in \mathcal{I}_\ell, j \in \mathcal{J}_\ell$ & unique identifiers of S and A &  $ \mathcal{E}$ & state of the world\\
$\Sigma_\ell: \mathcal{X} _\ell \cup \mathcal{I}_{\ell} \rightarrow \mathcal{V}$  & store of node  $\ell$ & $\!\!\!\sfreccia$ & reduction relation\\

$\widehat{G} = ({\mathbb N}, {\mathbb T}, Z, R)$ & regular tree grammar & 
         $i^{\ell}, v^{\ell}, f^\ell \in {\mathbb T}$ & abstract sensor, value, function\\
$Lang(\widehat{G})$ & language of $\widehat{G}$ & $\hat{t} \in \mathcal T$ & tree over ${\mathbb T}$ \\
$\hat{v} = (Z, R)$ & shorthand for $\widehat{G}$ & $\widehat{\mathcal V} = 2^{{\mathbb N} \times {\mathbb R}}$ & abstract terms   \\
$\mathbb R$ & the productions over ${\mathbb N}, {\mathbb T}$ & 
                $ \widehat{\Sigma}{_\ell}: {\mathcal X} \cup {\mathcal I}_\ell \rightarrow 2^{\widehat{\mathcal V}} $ & abstract store at $\ell$\\
$(\widehat{\Sigma}, \Theta)$ & estimate for terms $E$ &$ \widehat{\Sigma} = \bigcup_{\ell \in \mathcal L}\,\widehat{\Sigma}{_\ell} $ & abstract store \\
$(\widehat{\Sigma},\kappa,\Theta)$ & estimate for nodes $N$ & 
          $\kappa : {\mathcal L} \rightarrow {\mathcal L} \times \bigcup_{i = 1}^m \widehat{\mathcal V}^i$ & abstract network environment\\
$\vartheta \in  2^{\widehat{\mathcal V}}$ & set of approximated values & $\Theta: {\mathcal L} \rightarrow\ 2^{\widehat{\mathcal V}}$ & abstract data collection \\

$Comp(\ell_2,\ell_1)$ & compatibility function & 
          $\Sigma^{i}_{\ell} : \mathcal{X}_{\ell} \cup \mathcal{I} _\ell \rightarrow \mathcal{V} \times \widehat{\mathcal T}$ &
          instrumented store \\
$\dsem{E}^i_{\Sigma^{i}_{\ell}}$ & instrumented semantics & $\Sigma_\ell^i \bowtie \widehat{\Sigma}_\ell$ & concrete and abstract store agree\\

$k \in \mathcal K$ & crypto key & $enc^\ell_r, k_i$ & abstract values for security\\
$\mathfrak{D}$ & set of security tags & 
        $ \mathfrak{T_D} : \widehat{\mathcal{T}} \rightarrow \mathfrak{D}$ and &  tagging functions \\
$\mathfrak{P} \subseteq \mathfrak{D} \times \mathcal{L} \times \mathcal{L}$ & security policy &  
        $ \mathfrak{T_S} : \widehat{\mathcal{V}} \rightarrow \mathfrak{D}$ & \\
\hline
\end{tabular}
\end{center}
\label{tab:notation}
\caption{Notation and abbreviations}
\end{table}%
}
\normalsize

%% file: sections/proof.tex
% !TEX root = ../main_coord2016.tex
%

In this appendix, we restate the theorems presented earlier
in the paper and give the proofs of their correctness,  with the help of some auxiliary lemmata. 

Our first lemma guarantees that the analysis for expressions respects the instrumented denotational semantics.

\begin{lemma}[Subject reduction for expressions]\label{SR4E}
\ \\
For all $E$, $\nEFORM{E}{\vartheta}$ implies that
there exist $\Sigma{_\ell^i}$ and  $\widehat{G} \in \vartheta$ such that $\Sigma{_\ell^i} \bowtie \widehat{\Sigma}_\ell$ and $(\dsem{E}^i_{\Sigma^i_\ell})_{\downarrow_2} \in Lang(\widehat{G})$.
\end{lemma}
\begin{proof}
A straightforward induction on all the rules in \tablename~\ref{analysisT} and on rule (E-enc) suffices.
\end{proof}

\noindent Before proving the correctness of the analysis for systems of nodes, we prove that it is invariant under the structural congruence and is preserved under node components reduction.

\begin{lemma}[Congruence]\label{lemma:congr}\ \\
\begin{itemize}
\item
If $B \equiv B'$ then
$\nPAFORM{\ell}{B}$ iff $\nPAFORM{\ell}{B'}$.
\item
If $N \equiv N'$ then
$\nNAFORM{N}$ iff $\nNAFORM{N'}$.
\end{itemize}
\end{lemma}
\begin{proof} 
It suffices to inspect the rules for $\equiv$, since
associativity and commutativity of $\wedge$
reflects the same properties of both $|$ and $\|$, and to recall that any triple is a valid
estimate for $\NIL$.  
\end{proof}

\begin{lemma}[Subject reduction for node components]\label{SRtheoremB}\ \\
If $\nPAFORM{\ell}{B}$ and $B \sfreccia B'$ and $\Sigma_\ell^i \bowtie \widehat{\Sigma}_\ell$, for $\Sigma_\ell^i$ {\em ($in$ $B$)}
then $\nPAFORM{\ell}{B'}$ and $\Sigma{_\ell^i}' \bowtie \widehat{\Sigma}_\ell$, for $\Sigma{_\ell^i}'$ {\em ($in$ $B'$)}.
\end{lemma}
\begin{proof} 
Our proof is by induction on the shape of the derivation of $B \sfreccia B'$ and by cases on the last rule used.
In all the cases below we will have that for $\Sigma^{i}$ and $B_1$ components of $B$,  both facts (*) $\nPAFORM{\ell}{\Sigma^{i}}$, and (**) $\nPAFORM{\ell}{B_1}$ trivially hold.
So we will omit mentioning these judgements below, for the sake of brevity.
Also, we will implicitly apply the rules of Table~\ref{analysis}, e.g.\ the repeated applications of rule (B-par).
\\
\begin{itemize}
\item
 Case (Sense). We assume 
$\nPAFORM{\ell}{\Sigma^{i} \parallel probe(i).\, S \parallel B_1 }$ that has been proved because rule (B-sen) proves that
$\nPAFORM{\ell}{probe(i).\, S}$ holds.
By the same rule also $\nPAFORM{\ell}{S}$ holds, and thus
$\nPAFORM{\ell}{\Sigma^{i}\{(v,i^\ell)/{i}\} \parallel S \parallel B_1}$.
We are left to show that
$\Sigma{_{\ell}^i}' \bowtie \widehat{\Sigma}_\ell$.
By hypothesis $\Sigma{_{\ell}^i}' (y)= \Sigma{_{\ell}^i}(y)$ for all $y \in {\mathcal X}_{\ell} \cup {\mathcal I}_{\ell}$ such that $y \neq i$, 
while $(\Sigma^i_{\ell} (i))_{\downarrow_2} = i^\ell \in Lang(\widehat{G})$, for some $\widehat{G} \in \widehat{\Sigma}_{\ell}(i)$, which \mbox{trivially holds by (*)} and Lemma~\ref{SR4E}.

\item
 Case (Asgm). We assume 
$\nPAFORM{\ell}{\Sigma^{i} \  \|\  x:=E.\,P \  \|\ B_1 }$
that has been proved because the following conditions hold:         
\begin{eqnarray}
&& \nEFORM{E}{\vartheta} \label{c1}
\\
&&\forall \hat{v}:\; \hat{v} \in \vartheta\ \Rightarrow
\hat{v} \in \widehat{\Sigma}{_\ell}({x})\hspace*{4ex} \label{c2}\\
&& \nPAFORM{\ell}{P}\label{c3}
\end{eqnarray}
We have to prove that
$\nPAFORM{\ell}{\Sigma^{i}\{(v,\hat{v})/x\} \  \|\  P \  \|\ B_1 }$ that trivially holds because of (*); (\ref{c3}) and Lemma~\ref{SR4E}; and (**).
We have only to prove that
$\Sigma{_{\ell}^i}' \bowtie \widehat{\Sigma}_\ell$, by knowing that (by hypothesis) $\Sigma{_{\ell}^i} \bowtie \widehat{\Sigma}_\ell$, that implies
$\Sigma{_{\ell}^i}' (y)= \Sigma{_{\ell}^i}(y)$ for all $y \in {\mathcal X}_{\ell} \cup {\mathcal I}_{\ell}$ such that $y \neq x$.
Therefore $\Sigma{_{\ell}^i}' (y)= \Sigma{_{\ell}^i}(y)$ for all $y \in {\mathcal X}_{\ell} \cup {\mathcal I}_{\ell}$ such that $y \neq x$,
while $(\Sigma^i_{\ell} (x))_{\downarrow_2} \in  Lang(\hat{v} \in \widehat{\Sigma}_{\ell}(x))$ because of (\ref{c2}) and of Lemma~\ref{SR4E}.

\item
 Case (E-out). We assume 
$\nPAFORM{\ell}{\Sigma^{i}  \  \|\   \OUTM{E_1,\! \cdots \!,E_r}{L}. \, P \  \|\  B_1}$
that has been proved because the following conditions hold:       
\begin{eqnarray*}
&&  \bigwedge_{i=1}^r\; \nEFORM{E_i}{\vartheta_i} \\
&&  \  \nPAFORM{\ell}{P}
\\
&&  \  \forall \hat{v}_1,\cdots,\hat{v}_r:\; \bigwedge_{i=1}^r\, \hat{v}_i \in \vartheta_i\ \Rightarrow
\forall \ell' \in L:
         (\ell,  \mess{\hat{v}_1,\cdots,\hat{v}_r}) \in  \kappa(\ell')
\end{eqnarray*}
Proving that
$
\nPAFORM{\ell}{\OUTM{v_1,\cdots,v_r}{L}. \, \NIL  \ \| \ \ P \ \| \ B_1}
$
is straightforward because of the three above conditions and because of Lemma~\ref{SR4E}.

\item
Case (Decrypt). We assume \\
$  \nPAFORM{\ell}{\Sigma^i \ \| \  \DECSO{E}{E'_1,\cdots,E'_j}{x_{j+1},\cdots,x_r}{k}{}{P} \ \| \ B_1}  $
 that has been proved because the following conditions hold
 \begin{eqnarray}
&& \nEFORM{E}{\vartheta} \label{decr1}
\\
&& \bigwedge_{i=1}^r\; \nEFORM{E_i}{\vartheta_i}  \label{decr2}
\end{eqnarray}
and $ \forall \; \hat{v} \in \vartheta $ such that $\texttt{D}(\hat{v}, k) = [\hat{v}_1,\cdots, \hat{v}_r]$ that also implies the following two conditions
 \begin{eqnarray}
&& \bigwedge_{i=j+1}^r\;  \hat{v}_i \in \widehat{\Sigma}{_\ell}({x_i}) \label{decr3}
\\
&&  \nPAFORM{\ell}{P}   \label{decr4}
\end{eqnarray}
We have to prove that
$\nPAFORM{\ell}{\Sigma^{i} \  \|\  P[v_{j+1}/x_{j+1},\cdots,v_r/x_r] \  \|\ B_1 }$ that trivially holds because of (*); (\ref{decr1}), (\ref{decr2}), (\ref{decr4}) and Lemma~\ref{SR4E}; and (**).
We are left  to prove that
$\Sigma{_{\ell}^i}' \bowtie \widehat{\Sigma}_\ell$, by knowing that (by hypothesis) $\Sigma{_{\ell}^i} \bowtie \widehat{\Sigma}_\ell$, that implies
$\Sigma{_{\ell}^i}' (y)= \Sigma{_{\ell}^i}(y)$ for all $y \in {\mathcal X}_{\ell} \cup {\mathcal I}_{\ell}$ such that $y \neq x$.
Therefore $\Sigma{_{\ell}^i}' (y)= \Sigma{_{\ell}^i}(y)$ for all $y \in {\mathcal X}_{\ell} \cup {\mathcal I}_{\ell}$ such that $y \neq x$,
while $(\Sigma^i_{\ell} (x))_{\downarrow_2} \in  Lang(\hat{v} \in \widehat{\Sigma}_{\ell}(x))$ because of (\ref{decr1}), (\ref{decr2}) and because of Lemma~\ref{SR4E}.

\item
 The cases (Cond), (Int) and (A-com) are straightforward; 
the case (Act) is trivial because for us the world is a black box;
the case (ParB) directly follows from the induction hypothesis; 
Lemma \ref{lemma:congr} suffices to prove the case (CongrB).
\qedhere
\end{itemize}
\end{proof}

\subjectreduction*
\begin{proof} 
Our proof is by induction on the shape of the derivation of $N \sfreccia N'$ and by cases on the last rule used.
In all the cases below we will have that (*) $\nPAFORM{\ell}{\Sigma^{i}}$, as well as that (**) $\nPAFORM{\ell}{B_m}$ (whenever $m \geq 1$), for $\Sigma^i, B_m$ components of $N$.
So we will omit mentioning these judgements below, for the sake of brevity.
Also, we will implicitly apply the rules of Table~\ref{analysis}, e.g.\ the applications of rules (N-node) and (N-par).
\begin{itemize}

\item
Case (Multi-com). We assume
\[
\!\!\! \nNAFORM{\ell_1 \! : \!  [\OUTM{v_1,\! \cdots \!,v_r}{L}. \, \NIL \  \|\  B_1]\ | \
             \ell_2  \! : \! [\Sigma^{i}_{\ell_{2}} \ \| \ (E_1,\! \cdots \!,E_j;x_{j+1},\! \cdots \!,,x_r).Q \ \| \ B_2]}
             \] 
that is implied by
$\nNAFORM{ \ell_1: [\OUTM{v_1,\cdots,v_r}{L}. \, \NIL \  \|\  B_1]}$ and by\\
             $\nNAFORM{\ell_2:[\Sigma^{i}_{\ell_{2}} \ \| \ (E'_1,\cdots,E'_j;x_{j+1},\cdots,x_r).Q \ \| \ B_2]}$
that have been proved because the following conditions hold:      
 
\begin{eqnarray}
&& \bigwedge_{i=1}^r\; \nEFORMM{\ell_1}{v_i}{\vartheta_i} \label{c5}
\\   
&& \nPAFORM{\ell_1}{\NIL}\label{c7}
\\
&& 
\forall \hat{v}_1,\cdots,\hat{v}_r:\; \bigwedge_{i=1}^r\, \hat{v}_i \in \vartheta_i\ \Rightarrow
\forall \ell' \in L:
         (\ell_1,  \mess{\hat{v}_1,\cdots,\hat{v}_r}) \in  \kappa(\ell')  \hspace*{4ex} \label{c6}
\\
&& \bigwedge_{i=1}^{j}\;  \nEFORMM{\ell_2}{E_i}{\vartheta'_i} \label{c8}
\\
\end{eqnarray}

%&&\forall (\ell', \mess{ \hat{v} _1,\cdots,\hat{v}_k}) \in \kappa(\ell_2) :\;
%	  \bigwedge_{i=1}^{j}\;  \hat{v}_i \in \vartheta'_i\  \Rightarrow  \label{c9}
%	  \\[-4ex]

and $\forall (\ell_1, \mess{ \hat{v} _1,\cdots,\hat{v}_k}) \in \kappa(\ell_2)$ such that 
$Comp(\ell_1,\ell_2) \hspace{2mm}(\clubsuit)\label{c4}$

\begin{eqnarray}
%& & \forall (\ell', \mess{ \hat{v} _1,\cdots,\hat{v}_k}) \in \kappa(\ell_2)  \text{ s.t. }Comp(\ell_1,\ell_2)\label{c4} \\
& &  \bigwedge_{i=j+1}^r\;  \hat{v}_i \in \widehat{\Sigma}{_{\ell_{2}}}({x_i}) \label{c10} \ \\
& &  \nPAFORM{\ell_2}{Q} \label{c11}
\end{eqnarray}
Note that since~$\clubsuit$ holds, also \ref{c10} and~\ref{c11} do.
Also, $\forall i\; \nEFORMM{\ell_1}{v_i}{\vartheta_i}$ implies $\hat{v}_i \in \vartheta_i$, 
where $\hat{v}_i = (\dsem{v_i}^i_{\Sigma^{i}_{\ell_{2}}})_{\downarrow_2}$, and that $\ell_2 \in L$ because $N \sfreccia N'$.
%
%The premise of (Multi-com) ensures that
%$\ell_2 \in L$, $Comp(\ell_1,\ell_2)$, and $\bigwedge_{i=1}^{j} {(v_i, \hat{v}_i)}= \dsem{E_i}_{\Sigma_2}$, and therefore 
%$\bigwedge_{i=1}^j\ \hat{v_i} \in \vartheta'_i$ follows.
We have to prove that
\[
\!\!\! \nNAFORM{\ell_1: [\OUTM{v_1,\cdots,v_r}{L \setminus \{\ell_2\}}. \NIL   \|  B_1]
              \ | \ \ell_2: [\Sigma^{i}_{\ell_{2}}\{v_{j+1}/x_{j+1},\cdots,v_r/x_r\}  \|  Q \|  B_2]}
\]              
that in turn amounts to prove that
\[
\begin{array}{l}
(a)\ \nPAFORM{\ell_1}{\OUTM{v_1,\cdots,v_r}{L \setminus \{\ell_2\}}. \,\NIL  \ \| \ B_1 } 
\\
(b)\  \nPAFORM{\ell_2}{ \Sigma^{i}_{\ell_2}\{(v_{j+1}, \hat{v}_{j+1})/x_{j+1},\cdots,(v_r, \hat{v}_r)/x_r\} \ \| \ Q\ \| \ B_2 }
\end{array}
\]    
We have that $(a)$ holds trivially because of (\ref{c5}),~(\ref{c7}) and~(\ref{c6}) (of course $L\setminus \{ \ell_2\} \subseteq L$), while
$(b)$ holds because of the remaining conditions and of Lemma~\ref{SR4E}.
We are left to prove that
$\Sigma{_{\ell_{2}}^i}' \bowtie \widehat{\Sigma}_{\ell_{2}}$.
Now, we know that $\Sigma{_{\ell_{2}}^i}' (y)= \Sigma{_{\ell_{2}}^i}(y)$ for all $y \in {\mathcal X}_{\ell_{2}} \cup {\mathcal I}_{\ell_{2}}$ such that $y \neq x_i$.
The condition $(\Sigma^i_{\ell_2} (x_i))_{\downarrow_2} \in  Lang(\hat{v} \in \widehat{\Sigma}_{\ell_{2}}(x_i))$ holds for all $x_i$  
because of (\ref{c10}).

\item
The case (CongrN) follows from Lemma \ref{lemma:congr}; and the remaining cases directly follow from the induction hypothesis.
\qedhere
\end{itemize}
\end{proof}

\begin{definition}\label{partialorder}
The set of estimates can be partially ordered by setting \\
$(\widehat{\Sigma}_1,\kappa_1, \Theta_1)$ $\sqsubseteq (\widehat{\Sigma}_2,\kappa_2, \Theta_2)$
%%%$(\widehat{\Sigma}_1,\kappa_1, \Theta_1, \alpha_1)$ $\sqsubseteq (\widehat{\Sigma}_2,\kappa_2, \Theta_2, \alpha_2)$
iff
\begin{itemize}
\item
$\forall x \in {\mathcal X}_{\ell} \cup {\mathcal I}_{\ell}: \widehat{\Sigma}_1(x) \subseteq \widehat{\Sigma}_2(x)$

\item
$\forall \ell \in {\mathcal L}: \kappa_1(\ell) \subseteq \kappa_2(\ell)$

\item
$\forall \ell \in {\mathcal L}: \Theta_1(\ell) \subseteq \Theta_2(\ell)$

%%%\item
%%%$\forall j \in {\mathcal J}_\ell, \forall \ell \in {\mathcal L}: \alpha_1(\ell, j) \subseteq \alpha_2(\ell, j)$
\end{itemize}
\end{definition}

\noindent
This suffices for making the set of proposed
solutions into a complete lattice; 
we can thus write 
%$(\widehat{\Sigma}_1,\kappa_1, \Theta_1, \alpha_1) \sqcup (\widehat{\Sigma}_2,\kappa_2, \Theta_2, \alpha_2)$
$(\widehat{\Sigma}_1,\kappa_1, \Theta_1) \sqcup (\widehat{\Sigma}_2,\kappa_2, \Theta_2)$
for the binary least upper bound (defined point-wise),
$\sqcap{\mathcal M}$ for the greatest lower bound of a set ${\mathcal M}$
of proposed solutions (also defined pointwise),
and $(\bot,\bot,\bot)$ for the least element.

A Moore family ${\mathcal M}$ contains a greatest element ($\sqcap \emptyset$) and a least element ($\sqcap {\mathcal M}$). 
Since the set of analysis estimates constitutes a Moore family, we have that
there always is a least estimate to the specification in \tablename~\ref{analysis}.

%\begin{restatable}[Existence of estimates]{theorem}{extsolutions}\label{Mooretheorem}
%\  \\
%Given $N$ its estimates form a Moore family, hence it has a minimal element.
%\end{restatable}

\extsolutions*
\begin{proof}
%Let ${\mathcal M} = \{ (\widehat{\Sigma}_r,\kappa_r, \Theta_r, \alpha_r) \}$ be a set of estimates for $N$.
Let ${\mathcal M} = \{ (\widehat{\Sigma}_r,\kappa_r, \Theta_r) \}$ be a set of estimates for $N$.
We proceed by structural induction on $N$,
%We have to check that $\sqcap {\mathcal M} = (\widehat{\Sigma}',\kappa', \Theta', \alpha') \models N$.
to check that $\sqcap {\mathcal M} = (\widehat{\Sigma}',\kappa', \Theta') \models N$.
We just consider one case. The others are similar.

Case $\ell : [B]$.
Since 
%%$\forall k: (\widehat{\Sigma}_r,\kappa_r, \Theta_r, \alpha_r) \models  \ell : [B]$, then 
%%$(\widehat{\Sigma}_r,\kappa_r, \Theta_r, \alpha_r) \models_{_{\ell}}B$.
$\forall r: (\widehat{\Sigma}_r,\kappa_r, \Theta_r) \models  \ell : [B]$, then 
$(\widehat{\Sigma}_r,\kappa_r, \Theta_r) \models_{_{\ell}}B$.

\noindent
%Using the induction hypothesis and the fact that the components of $(\widehat{\Sigma}',\kappa', \Theta', \alpha')$ are  
Using the induction hypothesis and the fact that the components of $(\widehat{\Sigma}',\kappa', \Theta')$ are obtained pointwise, it follows that 
%%%$(\widehat{\Sigma}',\kappa', \Theta', \alpha') \models_{_{\ell}}B$
%%%thus establishing the required judgement
%%%$(\widehat{\Sigma}',\kappa', \Theta', \alpha') \models  \ell : [B]$.
$(\widehat{\Sigma}',\kappa', \Theta') \models_{_{\ell}}B$
thus establishing the required judgement
$(\widehat{\Sigma}',\kappa', \Theta') \models  \ell : [B]$.
\end{proof}

The following auxiliary lemma helps proving Corollary~\ref{cor:Theta}, and in its statement we let $B \xrightarrow{E_1, ...,E_n}_\ell B'$ denote a reduction in which all $E_i$ are evaluated at node $\ell$, omitting the environment for readability.

\begin{lemma}\label{cor:ThetaB}
If $\nPAFORM{\ell}{B}$ and $B \xrightarrow{E_1, ...,E_n}_\ell  B'$ then $\forall r \in [0,n]$ 
there exists $\widehat{G} \in \Theta(\ell)$ such that
$(\dsem{E_r}^i_{\Sigma^i_\ell})_{\downarrow_2} \in Lang(\widehat{G})$.
\end{lemma}

\begin{proof}
By Lemma~\ref{SRtheoremB}, $\nPAFORM{\ell}{B'}$ holds. 
The proof proceeds by induction on the shape of the derivation of $B \sfreccia B'$ and by cases on the last rule used.
\begin{itemize}
\item 
Case (Asgm).
If this rule is applied, then 
$B =  \Sigma^{i} \  \|\  x:=E.\,P \  \|\ B_1.
$
Since 
$
\nPAFORM{\ell}{B}
$
we have that, in particular, 
$\nPAFORM{\ell}{x : = E}.\,{P}$ and in turn that 
$ \nEFORM{E}{\vartheta}$, and Lemma~\ref{SR4E} suffices for establishing that $ \nEFORM{\dsem{E}^i}{\vartheta}$.
By the rules in \tablename~\ref{analysisT}, we have 
the required  $\vartheta \subseteq \Theta(\ell)$.

\item
Case (Decrypt). 
In this case we have that \\
$\nPAFORM{\ell}{B = \Sigma^i \ \| \  \DECSO{E}{E'_1,\cdots,E'_j}{x_{j+1},\cdots,x_r}{k}{}{P} \ \| \ B_1}$.
Since the estimate is valid, we have in particular that the premises of the rule (i.e.\ conditions (\ref{decr1}), (\ref{decr2}), and (\ref{decr3}) stated in the proof of Lemma~\ref{SRtheoremB}) hold.
The existence of the required $\widehat{G}$ now follows, because Lemma~\ref{SR4E} can be used to recover the abstract values $\hat{v}, \hat{v}_1, \dots, \hat{v}_n$ needed in the rule (P-dec).

\item The cases (Ev-out), and (Cond) are similar.

\item
The cases (ParB) and (CongrB) directly follow from the induction hypothesis. 

\item
No other rules evaluate terms and their proof is trivial.
\qedhere
\end{itemize}
\end{proof}

\corTheta*
\begin{proof}\ %
\begin{enumerate}
\item
Immediate; note that for the rule (Node) Lemma~\ref{cor:ThetaB} suffices.
\item
By Theorem~\ref{SRtheorem}, we have that $\nNAFORM{N'}$, so we proceed by
induction on the shape of the derivation of $N \sfreccia N'$ and by cases on the last rule used.
\begin{itemize}
\item
Case (Multi-com) follows directly from the subject reduction Theorem~\ref{SRtheorem}.
%If this rule is applied, than 
%$N$ is in the form 
%\[\!\!\! \nNAFORM{\ell_1 \! : \!  [\OUTM{v_1,\! \cdots \!,v_k}{L}. \, \NIL \  \|\  B_1]\ | \
%             \ell_2  \! : \! [\Sigma^{i}_2 \ \| \ (E_1,\! \cdots \!,E_j;x_{j+1},\! \cdots \!,,x_k).Q \ \| \ B_2]}
%             \] 
%Since 
%$
%\nNAFORM{N}
%$
%we have both 
%$\nPAFORM{\ell_1}{\OUTM{v_1,\! \cdots \!,v_k}{L}. \, \NIL}$ and  the required 
%$\forall \ell' \in L: (\ell_1,  \mess{\hat{v}_1,\cdots,\hat{v}_k}) \in  \kappa(\ell')$, 
%there exists $\hat{G} \in \Theta(\ell)$ such that
%$v_{i\downarrow_2} \in Lang(\hat{v_i})$.
%
%where $\hat{v}_i= v_{i_{\downarrow_2}}$, for all $i$.
\item
Cases (ParN), (CongrN), and (Node) directly follow from the induction hypothesis, and for the other rules  the premise is false.
\qedhere
\end{itemize}
%%%
%%%\item
%%%Suppose that $N \xrightarrow{\langle{j,\gamma}\rangle}_{\ell} N'$ then $\gamma \in \alpha_{\ell}(j)$.
%%%We have that $\nNAFORM{N'}$ by Theorem~\ref{SRtheorem}, so we proceed by
%%%induction on the shape of the derivation of $N \sfreccia N'$ and by cases on the last rule used.
%%%
%%%\begin{itemize}
%%%\item
%%%Case (A-com)
%%%We have a process with prefix  $\langle{j,\gamma}\rangle$ and therefore $\gamma \in \alpha_{\ell}(j)$, because $\nNAFORM{N}$.
%%%
%%%\item
%%%Cases (ParN), (CongrN), and (Node) directly follow from the induction hypothesis.
%%%
%%%\item
%%%In all the remainin cases the premise is false.
%%%\end{itemize}

\end{enumerate}
\end{proof}

\ingredient*
\begin{proof}
Immediate.
%%%%For proving the first item, suppose by contradiction that there exists a transition $N' \rightarrow N''$, derived by using the rule  \mbox{(A-com)}, with action $\gamma$ and actuator $j$. 
%%%%Then a control process in $N'$ must include the prefix $\langle{j,\gamma}\rangle$. 
%%%%Now, because of Theorem~\ref{SRtheorem}, we have that $\nNAFORM{N}$ implies $\nNAFORM{N'}$. 
%%%%As a consequence, the analysis for $N'$ must hold for the process $\OUTS{j,\gamma}{P}$ and therefore $\gamma \in \alpha_{\ell}(j)$: contradiction.
%%%%
%%%%The second item  follows by a similar argument.
\end{proof}

\propaga*
\begin{proof}
Suppose that both items (1) and (2) hold and assume by contradiction that $N$ enjoys $\mathfrak{P}$ with respect to $\mathfrak{D}$, i.e.\ that  
$N \rightarrow^* N' \xrightarrow{\mess{v_1,\dots,v_r}}_{\ell_1, \ell_2} N''$
and that $\exists \bar{i} : \mathfrak{P}(\mathfrak{T_D}(v_{\bar{i} \downarrow 2}), \ell_1, \ell_2)$ does not hold.
Then, by Corollary~\ref{cor:Theta},
it holds $(\ell_1,\mess{\hat{v}_1,\dots,\hat{v}_r}) \in \kappa(\ell_2)$, where for all $i$ there exists 
$\widehat{G} \in \hat{v}_i$ such that  $v_{i\downarrow_2} \in Lang(\widehat{G})$.
Now, since item (2) holds, we know that $\forall i. \ \mathfrak{P}(\mathfrak{T_S}(\hat{v}_i), \ell_1, \ell_2)$ holds.
Since $\mathfrak{T_D}, \mathfrak{T_S}$ is a pair of tagging functions, for all $\hat{t} \in Lang(\widehat{G})$
we have that $\mathfrak{P}(\mathfrak{T_S}(\hat{v}_i), \ell_1, \ell_2) = \mathfrak{P}(\mathfrak{T_D}(\hat{t}), \ell_1, \ell_2)$, in particular for $i = \bar{i}$: contradiction.
\end{proof}

\sed*
\begin{proof}
Immediate by inducing on the structure of $\hat{t}$.
\end{proof}

\confine*
\begin{proof}
Follows easily from Theorem~\ref{Th:propagated} by letting $\mathfrak{D} = \{secret, public\}$; 
$ \mathscr{S}, \mathscr{P}$ as pair of tagging functions (Lemma~\ref{lemma:S-e-D}); and 
$\mathfrak{P}(d, \ell_1, \ell_2) = true$ if and only if $d = public$.
\end{proof}

\nrunwd*
\begin{proof}
Follows easily from Theorem~\ref{Th:propagated} by letting $\mathfrak{D} = \{\bullet\}$; 
$ \mathfrak{T_S}, \mathfrak{T_D}$ as pair of (constant) tagging functions; and 
$\mathfrak{P}(d, \ell_1, \ell_2) = true$ if and only if $level(\ell_1) \leq level(\ell_2)$.
\end{proof}

\ceo*
\begin{proof}
Immediate by inducing on the structure of $\hat{t}$.
\end{proof}

\selectiveprop*
\begin{proof}
Follows easily from Theorem~\ref{Th:propagated} by letting $\mathfrak{D} = \{\mathit{confined}, open\}$;
$ \mathfrak{T_S} = \mathcal{O}$ and $\mathfrak{T_D} = \mathcal{C}$ that form a pair of tagging functions by Lemma~\ref{lemma:C-e-O}; and 
$\mathfrak{P}(d, \ell_1, \ell_2) = true$ if and only if $d = \mathit{confined}$ and $\ell_1, \ell_2 \in \tilde{\mathcal{L}}$.
\end{proof}